\documentclass{llncs}

\usepackage{cmap}
\usepackage[utf8]{inputenc}
\usepackage{graphicx}
\usepackage{microtype}
\usepackage{amsmath,amssymb}
\usepackage{url}
\usepackage{paralist}
\usepackage{booktabs}
\usepackage{array}
\usepackage{xcolor}
\usepackage{tikz}
\usepackage[pdfborder={0 0 0}]{hyperref}
\usepackage{ifthen}
\usepackage{wrapfig}
\usepackage{comment,todonotes}
\usetikzlibrary{matrix}

\usepackage{xspace}
\usepackage[english]{babel}
\makeatletter
\def\imod#1{\allowbreak\mkern10mu({\operator@font mod}\,\,#1)}
\makeatother

\newcommand{\fotwo}{\rm{FO}^2}

\newcommand{\kw}[1]{{\mathsf{#1}}\xspace}

\newcommand{\parentof}{\kw{ParentOf}}
\newcommand{\ancesof}{\kw{AncOf}}
\newcommand{\immleftsib}{\kw{LeftSibOf}}
\newcommand{\leftsib}{\kw{LeftOf}}

\newcommand{\ancof}{\ancesof}
\newcommand{\descof}{\kw{DescOf}}
\newcommand{\incomp}{\kw{InComp}}
\newcommand{\childof}{\kw{ChildOf}}

\newcommand{\nextsib}{\kw{NextSib}}

\newcommand{\vfull}{V_{\textit{full}}}
\newcommand{\vnodesc}{V_{\textit{noAncOf}}}
\newcommand{\vdesc}{V_{\textit{ancOf}}}
\newcommand{\vnochild}{V_{\textit{noParOf}}}
\newcommand{\vchild}{V_{\textit{parOf}}}

\newcommand{\myparagraph}[1]{{\bf #1.}}
\newcommand{\tl}{TL_{\textit{tree}}}
\newcommand{\utl}{UTL_{\textit{tree}}}

\newcommand{\navxp}{\kw{NavXP}}
\newcommand{\dsfnavxp}{\kw{DownSF\mbox{-}NavXP}}
\newcommand{\myeat}[1]{}

\newcommand{\exptime}{\kw{EXPTIME}}
\newcommand{\nexptime}{\kw{NEXPTIME}}
\newcommand{\nexp}{\nexptime}
\newcommand{\expspace}{\kw{EXPSPACE}}

\newcommand{\twoexp}{\kw{2EXPTIME}}

\newcommand{\pspace}{\kw{PSPACE}}

\newcommand{\up}{\kw{Update}}

\newcommand{\ch}{\kw{CH}}
\newcommand{\ns}{\kw{NS}}

\newcommand{\childax}{\kw{child}}

\newcommand{\selfax}{\kw{self}}
\newcommand{\descax}{\kw{descendant}}
\newcommand{\dosax}{\kw{descendant\mbox{-}or\mbox{-}self}}

\newcommand{\ancosax}{\kw{ancestor\mbox{-}or\mbox{-}self}}
\newcommand{\nextsibax}{\kw{next\mbox{-}sibling}}
\newcommand{\folsibax}{\kw{following\mbox{-}sibling}}
\newcommand{\prevsibax}{\kw{previous\mbox{-}sibling}}
\newcommand{\precsibax}{\kw{preceding\mbox{-}sibling}}

\newcommand{\sembrack}[1]{[\![#1]\!] }
\newcommand{\tp}{\kw{Tp}}
\newcommand{\desctype}{\kw{DescTypes}}
\newcommand{\sdesctype}{\kw{SelectedDescTypes}}
\newcommand{\subtree}{\kw{SubTree}}
\newcommand{\comb}{\kw{Val}}
\newcommand{\contexttype}{\kw{IncompTypes}}
\newcommand{\equivfull}{\equiv_{\kw{Full}}}
\newcommand{\anctype}{\kw{AncTypes}}

\newcommand{\bottomn}{\text{Low}}
\newcommand{\topn}{\text{High}}

\newcommand{\zerox}{\kw{ZeroX}}
\newcommand{\onex}{\kw{OneX}}
\newcommand{\zeroy}{\kw{ZeroY}}
\newcommand{\oney}{\kw{OneY}}

\makeatletter
\def\@Opargbegintheorem#1#2#3#4{#4\trivlist
      \item[\hskip\labelsep{#3#1}]{#3#2\@thmcounterend\ }}

\newcommand{\definerep}[2]{%
\spnewtheorem*{#1rp}{#2}{\bf}{\itshape}
\newenvironment{#1rep}[2]{%
  \ifthenelse{\equal{##1}{*}}
  {\begin{#1rp}[\ref{##2}]}
  {\begin{#1}\label{##2}}}
{\ifthenelse{\equal{\@currenvir}{#1}}{\end{#1}}{\end{#1rp}}}
}
\definerep{proposition}{Proposition}
\definerep{lemma}{Lemma}
\definerep{theorem}{Theorem}

\def\@spthm#1#2#3#4{\topsep 3\p@ \@plus1\p@ \@minus1\p@
\refstepcounter{#1}%
\@ifnextchar[{\@spythm{#1}{#2}{#3}{#4}}{\@spxthm{#1}{#2}{#3}{#4}}}
\def\@Thm#1#2#3{\topsep 3\p@ \@plus1\p@ \@minus1\p@
\@ifnextchar[{\@Ythm{#1}{#2}{#3}}{\@Xthm{#1}{#2}{#3}}}
\renewcommand\paragraph{\@startsection{paragraph}{4}{\z@}%
                       {-4\p@ \@plus -2\p@ \@minus -2\p@}%
                       {-0.2em \@plus -0.1em \@minus -0.1em}%
                       {\normalfont\normalsize\itshape}}
\renewcommand\section{\@startsection{section}{1}{\z@}%
                       {-8\p@ \@plus -2\p@ \@minus -2\p@}%
                       {4\p@ \@plus 2\p@ \@minus 2\p@}%
                       {\normalfont\large\bfseries\boldmath
                        \rightskip=\z@ \@plus 8em\pretolerance=10000 }}
\renewcommand\subsection{\@startsection{subsection}{2}{\z@}%
                       {-8\p@ \@plus -2\p@ \@minus -2\p@}%
                       {4\p@ \@plus 2\p@ \@minus 2\p@}%
                       {\normalfont\normalsize\bfseries\boldmath
                        \rightskip=\z@ \@plus 8em\pretolerance=10000 }}
\makeatother

\pdfmapfile{+ltlfonts.map}

\usepackage{setspace}

\title{Controlling the Depth, Size, and Number of Subtrees for Two-variable Logic on Trees}
\author{Saguy Benaim, Michael Benedikt,  Rastislav Lenhardt, and James Worrell}
\institute{Department of Computer Science, University of Oxford, UK }

\usepackage{verbatim}
\usepackage{graphicx}
\usepackage{ltlfonts}
\usepackage{setspace}
\usepackage{tikz} 
\usetikzlibrary{arrows}
\usetikzlibrary{automata}
\usetikzlibrary{shapes}
\usetikzlibrary{decorations.pathmorphing}
\usetikzlibrary{decorations.markings}
\pdfmapfile{ltlfonts.map}

\begin{document}
\maketitle
\begin{abstract}

Verification of properties of first order logic with two variables $\fotwo$ has been investigated
in a number of contexts.
Over
arbitrary structures it is known to be decidable with $\nexptime$ complexity, with finitely satisfiable
formulas having  exponential-sized models. 
Over word structures, where $\fotwo$  is known to have the same expressiveness
as unary temporal logic, the same properties hold.
Over finite labelled ordered trees $\fotwo$ is also of interest: it is known to have the same expressiveness
as navigational XPath, a common query language for XML documents. Prior work on
XPath and $\fotwo$ gives a $\twoexp$ bound for satisfiability of $\fotwo$.
In this work we give the first in-depth look at the complexity
of $\fotwo$ on trees, and on the size and depth of models.  We show that the doubly-exponential bound is not tight, and neither do
the $\nexptime$-completeness results from the word case carry over: the exact complexity
varies depending on the vocabulary used, the presence  or absence of a~schema, and
the encoding used for labels. Our results depend on an analysis of subformula types
in models of $\fotwo$ formulas, including techniques for controlling the number of distinct subtrees,
the  depth, and the size of a witness to finite satisfiability for $\fotwo$ sentences over trees.

\end{abstract}

\section{Introduction} \label{sec:intro}
The complexity of verifying properties over a class of structures depends on both
the specification language for  properties and the class of structures.
Full first-order logic (FO) has non-elementary complexity even when
applied to very restricted structures -- e.g. words.
The two-variable fragment of FO, $\fotwo $, is known to have better properties. Satisfiability
over
 arbitrary relational vocabularies is decidable, and satisfiable sentences have exponential-sized models  \cite{kolaitisvardi}.
Over words witness models can also be taken to be exponential, and
 the satisfiability problem is known to be  $\nexptime$-complete, as it is over general structures \cite{fo2_utl}. The satisfiability
results over words extend to give bounds on many related verification problems
\cite{BLW2}.

The $\nexptime$-completeness of $\fotwo$ over both general structures and word structures raises the
question of the impact of structural restrictions on analysis problems for {$\fotwo$}. Surprisingly
the complexity of satisfiability for  {$\fotwo$} on a~class of structures
satisfying a very simple graph-theoretic restriction -- namely, finite trees --
has not been investigated in detail.
 {$\fotwo$} over trees is known to correspond precisely to the navigational core of the
XML query language XPath \cite{derijkemarx}, and  the satisfiability problem for XPath 
is known to be complete for $\exptime$; given that the translation from $\fotwo$ to XPath is known
to be exponential \cite{derijkemarx}, this gives a $\twoexp$ bound on satisfiability for {$\fotwo$} over trees.

In this work we will consider the satisfiability problem for $\fotwo$ over finite trees, and the corresponding
question of the size and depth needed for witness models.
In particular, we will consider:
\begin{compactitem}
\item satisfiability in the presence of all navigational predicates  --
predicates for the parent/child relation, its transitive closure the 
descendant relation, the left- and right- sibling relations and their transitive closures
\item the impact on the complexity of limiting sentences
to make use of predicates in
a particular subset.
\item satisfiability over general unranked trees, and satisfiability in the presence of a schema
\item satisfiability over trees where nodes labels are denoted with explicit  unary labels
versus the case where node labels are boolean combinations over a~propositional alphabet
\end{compactitem}

We will show that each of these variations impacts the complexity of the problem.
In the process, we will show that the tree case differs in a number of important
ways from that of words. First, the complexity of satisfiability no longer matches
that of $\fotwo$ on general structures -- it is $\expspace$-complete.
Secondly, the basic technique for analyzing $\fotwo$ on words 
\cite{fo2_utl}--
bounds on the number of quantifier-rank types that occur in a structure -- is not useful for getting tight bounds
on $\fotwo$ over trees. Instead we will use a combination of methods, including
reductions to XPath, bounds on the number of subformula-based types, and a quotient construction
that is based not only on types, but on a set of distinguished witness nodes.
These techniques allow us to distinguish situations where satisfiable $\fotwo$-formulas have  models
of (reasonably) small depth, and situations where they have models of small size. This allows
us to get a  full picture of the complexity of $\fotwo$ satisfiability problems
on trees.

\myparagraph{Related work} Two-variable logic on   \emph{data trees}  -- trees where nodes
are associated with values in an infinite set--
has been
studied by Bojanczyk et. al. \cite{datatrees1}: there the main result
is decidability over the signature with data equality and the child relation.
Figueira's manuscript \cite{twosuccs} considers two-variable logic with
the successor relations corresponding to two linear orders, which is quite different from considering the
two successor relations derived from a tree order. Kieronski et. al. show that two-variable logic over two transitive
relations is undecidable. The complexity of two-variable logic over ordinary trees 
is explicitly studied only in \cite{leashed}, where it is (incorrectly, as we show) stated
that the complexity of satisfiability remains in $\nexptime$ for full two-variable logic.

%As an application of the satisfiability results, we will give new bounds on probabilistic
%XML documents.

{\bf Organization:}
Section \ref{sec:prelim} gives preliminaries.
Section \ref{sec:full} gives precise bounds for  the satisfiability  of full $\fotwo$ on trees.
Section \ref{sec:nochild} considers the case where the child predicate is absent,
while Section \ref{sec:nodesc} considers the case where the descendant predicate is absent.
%Section \ref{sec:probxml} gives applications to probabilistic document schemas.
Section \ref{sec:conc} gives conclusions.

\section{Logics and Models} \label{sec:prelim}

We will always use the term ``tree'' to denote a
\emph{finite ordered labelled tree}, where the labels are sets
of unary predicates $P_1 \ldots P_n$. An ordered tree will consist
of a~finite set of nodes, a directed edge relation $\parentof$ between nodes such that
the underlying graph forms a~tree in the usual sense, a mapping
of each $P_i$ to a~subset of the nodes, and a sibling relation $\nextsib$ between nodes
that forms the successor relation of a linear order when restricted to
the set of children of a given node. We sometimes write
$m ~ \descof ~n$ to denote that 
node $m$ is a~descendant of node $n$ in a tree, and similarly
write $m ~ \childof ~ n$ to denote that $m$ is a child of $n$.
A tree satisfies the \emph{unary alphabet restriction} (UAR) if exactly one $P_i$ holds of each node;
in such a tree the labels are just predicates.
Given a tree $t$ and node $n$, $\subtree(t,n)$ denotes the subtree of $t$ rooted at $n$.

We consider first-order logic sentences in which every subformula has at most two variables,
allowing the equality predicate as well as relations
from the following signatures for trees:
\begin{compactitem}
\item for general ordered  trees, we consider by default a signature $\vfull$ containing
predicates for the node predicates $P_i$,  as well as for the $\parentof$ relation,  its transitive
closure $\ancof$,
 the  $\immleftsib$ relation that holds of $c$ and $d$ if 
$c$ is the immediate left sibling of $d$, and its transitive closure
$\leftsib$. 
\item we let $\vnodesc$ be the vocabulary obtained by removing the descendant relation,
$\vchild$ be the vocabulary obtained by removing all binary relations other than $\parentof$,
$\vnochild$ be the vocabulary obtained by removing the $\parentof$ relation, and
$\vdesc$ be the vocabulary obtained by removing all binary relations other than $\ancof$.
%\item for $k$-ranked trees we consider the vocabulary $\vranked_k$ containing
%the relation $\parentof_i$ connecting a node to its $i^{th}$ child for each $i \leq k$, in
%addition to the descendant relation.
\end{compactitem}
We consider $k$-ranked trees as a particular class of unranked trees, and thus can ask whether
an $\fotwo$ sentence in any of the signatures above is true on a ranked tree.
Note that for $k$-ranked trees it is natural to consider signatures that include
the relation $\parentof_i$, connecting a node to its $i^{th}$ child for each $i \leq k$, either
in place of or in addition to the predicates above.
We will not consider a separate signature for ranked trees, since
it is easy to derive tight bounds
for ranked trees for such signatures based on the techniques introduced here.
Although we allow equality in our upper
bounds, it will not play any role in the lower bounds.
%CHECK!!!

The signatures above used predicates for which the first argument
is either higher up in the tree than the second argument ($\parentof(c,d)$
means that $c$ is the parent of $d$) or to the left of the
second argument.
However, in first-order logic, as well as in two-variable first-order logic,
we can express the inverse
of any atomic relation  as a formula. Thus we
can  use  formulas $x ~ \descof  ~ y$, $x ~ \childof ~ y$, etc.  with the obvious
meaning (e.g.  $x ~ \descof  ~ y$ meaning $\ancof(y,x)$).

For any vocabulary $V$ above, we let $\fotwo(V)$ denote the fragment of first-order logic
consisting of formulas such that  every subformula uses at most two variables. When $V$ is omitted
it is assumed to be $\vfull$.

A \emph{ranked tree schema} consists of a  bottom-up tree automaton on trees of some
rank $k$
 \cite{thomashandbook3}.   A tree automaton takes 
trees labeled from a finite set $\Sigma$.
We will  thus identify the  symbols in $\Sigma$ with predicates $P_i$, and
thus all trees satisfying the schema will satisfy 
 the UAR. 

We consider the following problems:
\begin{compactitem}
\item Given an $\fotwo$ sentence $\varphi$ and a schema $S$, determine
whether $\varphi$ is satisfied by some tree satisfying $S$. We consider
the combined complexity in the formula and schema.
\item Given an $\fotwo$ sentence $\varphi$, determine if there is
some tree (resp. $k$-ranked, unary alphabet tree) that
satisfies it.
\end{compactitem}

Some of our results will go through XPath,  a common language used for querying XML documents viewed as trees.
The navigational core of XPath is a modal language, analogous
to unary temporal logic on trees,  denoted $\navxp$.
%It is equivalent in expressive power to $\fotwo$
% but  
%exponentially less succinct than $\fotwo$ \cite{leashed}.
$\navxp$ is built on binary modalities, referred to as \emph{axis relations}. We will focus on the following axes:
$\selfax$, $\childax$, $\descax$, $\dosax$, $\ancosax$,  $\nextsibax$, $\folsibax$, $\precsibax$, $\prevsibax$.
In a tree $t$,  we associate each axis $a$ with a set $R^t_a$ of pairs of nodes.
 $R^t_{\childax}$  denotes the set of pairs of nodes $(x,y)$ in $t$ where $y$ is a child of $x$, and similarly
for the other axes
(see \cite{marx04}).

$\navxp$ consists of path expressions, which denote binary relations between nodes in a tree,
and filters, denoting unary relations.
Below we give the  syntax (from \cite{leashed}), using  $p$ to range over path expressions and $q$ over filters. $L$~ranges 
over symbols for each labelling of a node (i.e. for general trees, 
boolean combinations of predicates $P_1 \ldots P_n$, for UAR trees a single predicate).
\begin{align*}
p &::= step \;|\; p/p \;|\; p \cup p  &
step &::= axis \;|\; step[q] \\
q &::= p \;|\; lab()=L \;|\; q \wedge q \;|\; q \vee q \;|\; \neg q
\end{align*}
where axis relations are given above.

The semantics of $\navxp$ path expressions relative to a tree $t$ is given by:
\begin{inparaenum}
\item 
$\sembrack{\text{axis}} =  R^t_{\text{axis}}$ 
\item $\sembrack{\text{step}[q]} = \{ (n,n') \in \sembrack{\text{step}} \; : \;
n' \in \sembrack{q}     \}$
\item 
$\sembrack{p_1/p_2} = \{(n,n') \;:\; \exists w  (n,w) \in \sembrack{p_1} \wedge  (w,v) \in \sembrack{p_2}\}$
\item 
$\sembrack{p_1 \cup p_2} = \sembrack{p_1} \cup \sembrack{p_2}$.
\end{inparaenum}
%\end{compactitem}

For filters we have:
\begin{inparaenum}
%\begin{align*}
\item 
$\sembrack{lab()=L} = \{n :  n\mbox{ has label }L\}$ 
\item 
$\sembrack{p} = \{n: \exists n' ~ (n,n') \in \sembrack{p} \}$
\item 
$\sembrack{q_1 \wedge q_2} = \sembrack{q_1} \cap \sembrack{q_2}$
\item 
$\sembrack{\neg q}(n) = \{n:  n \not \in \sembrack{q}\}$. 
%\end{align*}
\end{inparaenum}
A $\navxp$ filter is said to hold of a tree $t$ if it holds of the root under the above semantics.

Marx and De Rijke showed an expressive equivalence of $\navxp$ and $\fotwo$,   extending
the  translation to Unary Temporal Logic in the word case:

%\begin{proposition} \label{tl2fo}
%There is a polynomial translation from $\utl$ to $\fotwo$ and
%from $\tl$ to $\fotwo[\vdesc]$.
%\end{proposition}

\begin{proposition} \label{prop:fo2tl} \cite{derijkemarx}
There is an exponential translation from $\fotwo$ to $\navxp$ with all axis and
from $\fotwo[\vdesc]$ to  $\navxp$ with only the descendant and ancestor axes.
\end{proposition}

Marx has shown that $\navxp$ has an exponential time satisfiability problem \cite{marx04}.
From this and the above proposition, we  get the following (implicit in \cite{derijkemarx}):

\begin{corollary} \label{naiveupper}
The satisfiability problem for $\fotwo$ is in $\twoexp$.
\end{corollary}

\section{Satisfiability for full $\fotwo$} \label{sec:full}

\myparagraph{Subformula types and exponential depth bounds}
In the analysis of satisfiability of $\fotwo$ for words of Etessami, Vardi,
and Wilke \cite{fo2_utl}, a $\nexp$ bound is achieved
by showing that any sentence with a finite model has a model of at most exponential size.
The small model property follows, roughly speaking, from
the fact that  \emph{any} model realizes only exponentially many ``quantifier-rank types'' 
--  maximal consistent sets of formulas of a given quantifier rank -- and  the fact that two nodes with the
same quantifier-rank type can be identified. % giving an exponential-sized model.

In the case of trees, this approach breaks down in several places.
It is easy to see that one cannot always obtain an exponential-sized model, since
a sentence can enforce binary branching and exponential depth.
Because there are doubly-exponentially many non-isomorphic small-depth subtrees,
there can be doubly-exponentially  many quantifier-rank types realized 
%in a tree
%, and even doubly-exponentially
even along a single path in a tree: so quantifier-rank types can not be used even to show an exponential depth bound.
We thus use \emph{subformula types} 
of a given $\fotwo$-formula $\varphi$ (for short, $\varphi$-types) -- these
are  maximal consistent collections of one-variable subformulas of $\varphi$. 
The $\varphi$-type of a node $n$ in a tree, $\tp_\varphi(n)$,  is defined
as the set of subformulas of $\varphi$ it satisfies.
The number
of $\varphi$-types  is only
exponential in $|\varphi|$, but  subformula types are more delicate than quantifier-rank types.
E.g. 
nodes with the same $\varphi$-type cannot always be identified without
changing the truth of $\varphi$. 
Most of the upper bounds  will be concerned
with handling this issue, by adding additional conditions on nodes to be identified,
and/or preserving additional parts of the tree.

\myparagraph{Upper bounds for $\fotwo$} We
exhibit the issues arising and techniques used to solve
them  by giving an upper bound for the full logic, $\fotwo$, which improves
on the $\twoexp$ bound one obtains via translation to modal logic. 

\begin{theorem}
\label{fullup}
The satisfiability problem for $\fotwo$ is in $\expspace$.
\end{theorem}

The key to the proof is to show the ``exponential depth property'':
\begin{lemma}
\label{depthlemma}
Every satisfiable $\fotwo$ sentence $\varphi$ has  a model $T'$ where the depth 
 is bounded by  $2^{poly(|\varphi|)}$, and
similarly for satisfiability w.r.t UAR trees or ranked schemas. The outdegree of nodes
can also be bounded by  $2^{poly(|\varphi|)}$.
\end{lemma}

We give the argument for the depth bound, leaving the similar proof for the branching bound to the appendix.
Given a tree $t$ and nodes $n_0$ and $n_1$ in $t$ with $n_1$ not an ancestor of $n_0$,
the \emph{overwrite} of $n_0$ by $n_1$ in $t$  is the tree $t(n_1 \rightarrow n_0)$ formed by
replacing 
%g every node that is a descendant of $n_0$ other than $n_1$ and its descendants, and
%  the subtree of  $n_0$ with $n_0$, thus 
the subtree of $n_0$ %in $t(n_1 \rightarrow n_0)$ 
with the subtree of $n_1$ in $t$. 
Let $F$ be the binary relation relating a node $m$ in $t$ to its copies in $t(n_1 \rightarrow n_0)$:
$n_1$ and its descendants have a single copy if $n_1$ is a descendant of $n_0$, and two copies otherwise;
nodes in  $\subtree(t,n_0)$ that are not in $\subtree(t,n_1)$ have no copies, and other nodes have a single copy.
In the case that $n_1$ is a descendant of $n_0$, $F$ is  a partial function.
We say an equivalence relation $\equiv$ on nodes of a tree $t$ is 
 \emph{globally $\varphi$-preserving} 
if for any equivalent nodes $n_0, n_1$ in $t$ with $n_0 \not \in \subtree(t,n_1)$,
%not an ancestor of $n_0$,
the $\varphi$-type of a node $n$ in $t$ is the same as the $\varphi$-type of
  nodes in $F(n)$ within $t(n_1 \rightarrow n_0)$.
We say it is \emph{pathwise $\varphi$-preserving}
if this holds for any node $n_0, n_1$ in $t$ with $n_1$ a descendant of $n_0$.
The \emph{path-index} of an equivalence relation on $t$ is the maximum of the number of equivalence classes represented
on any path, while the \emph{index} is the total number of classes.

We can not always overwrite a  node with another having the same $\varphi$-type,
but by adding additional information, we can get a pathwise $\varphi$-preserving relation with small path-index.
For a node $n$, let $\desctype(n)$ be the set of $\varphi$-types of descendants of $n$, and
$\anctype(n)$ the set of $\varphi$-types of ancestors of $n$. Let
$\contexttype(n)$ be the $\varphi$-types of nodes $n'$ that %are descendants of ancestors of $n$, but which 
are neither descendants
nor ancestors of $n$.
Say $n_0 \equivfull n_1$ if
they agree on their $\varphi$-type, %their labels,
%the $\varphi$-types of their parents, 
the set $\desctype$, the set $\anctype$, and the set $\contexttype$.

\begin{lemma} \label{collapse}
The relation $\equivfull$ is pathwise $\varphi$-preserving, and its path index is bounded by  $2^{poly(|\varphi|)}$.
Thus, there is a polynomial $P$ such that
for any tree $t$ satisfying $\varphi$ and root-to-leaf path $p$ of length at least
 $2^{P(|\varphi|)}$,  there are two nodes $n_0,n_1$ on $p$ such
that $t(n_1 \rightarrow n_0)$ still satisfies $\varphi$. Given a tree automaton $A$,
it can be arranged that $A$ reaches the same state on $n_0$ as on $n_1$.
\end{lemma}
Given Lemma \ref{collapse}, Lemma \ref{depthlemma} follows by  
contracting all paths exceeding a given length
until the depth of the tree is exponential in $|\varphi|$.
In fact  (e.g., for ranked trees)
$\equivfull$ can be used as the state set of a tree automaton.
The path index property implies that
the automaton goes through only exponentially many states on any path of a tree. By taking
the product of this automaton with a ranked schema, the corresponding depth bound relative to a schema
follows.

We give the simple argument for the path index bound in
Lemma \ref{collapse}, leaving the proof that $\equivfull$ is pathwise
$\varphi$-preserving to the appendix.  First, note that the total
number of $\varphi$-types is exponential in $|\varphi|$.  Now the sets
$\desctype(n)$ either become smaller or stay the same as $n$ varies
down a path, and hence can only change exponentially often. Similarly
the sets $\contexttype(n)$ and $\anctype(n)$ grow bigger or stay the same, and thus can
change only exponentially often. In intervals along a path where both
of these sets are stable, the number of possibilities for the
$\varphi$-type of a node is exponential. This gives the path index
bound.

Theorem \ref{fullup}  follows from combining Lemma \ref{depthlemma} with the following result on satisfiability 
of $\navxp$:
\begin{theorem}
\label{navxpbound}
The satisfiability of a $\navxp$ filter $\varphi$ over trees
of bounded depth~$b$ is in $\pspace$ (in $b$ and $|\varphi|$).
\end{theorem}
The result is proved in the appendix, but it is a variant of a result
from \cite{BFG} that finite satisfiability for the fragment of
$\navxp$ which contains only axis relations child, parent,
next-sibling, preceding-sibling, previous-sibling and
following-sibling is in $\pspace$.  Given Theorem \ref{navxpbound} we
complete the proof of Theorem \ref{fullup} by translating an $\fotwo$
sentence $\varphi$ into an $\navxp$ filter $\varphi'$ with an
exponential blow-up, using Proposition \ref{prop:fo2tl}.  By
Lemma \ref{depthlemma}, the depth of a witness structure is bounded by
an exponential in $|\varphi|$, and the $\expspace$ result follows.

\myparagraph{Lower bound} We now show a matching lower bound
for  the satisfiability problem.
%le $M = (Q, \Sigma, \delta, q_0, v)$ where:
%\begin{itemize}
%\item $Q$ is a finite set of states
%\item $\Sigma$ is finite alphabet
%\item $\delta: Q \times \Sigma \rightarrow \p(Q \times \Sigma \times \{L, R\})$ is a transition function
%\item $q_0$ is the initial state
%\item $v: Q \rightarrow \{\wedge, \vee, accept, reject\}$ is a type labeling function
%\end{itemize}
% Define a configuration
%If $M$ is in an accepting state ($v(q) = accept$) then the configuration is accepting, if it is in a rejecting state, it is rejecting. 
%A configuration in a $\wedge$ state  is accepting if all configurations reachable in one step (via $\delta$) are accepting, and rejecting if some configuration is rejecting. A configuration in a $\vee$ state is accepting if some configuration reachable in one step is accepting, and rejecting if all such configurations are rejecting. A word $w$ is accepted by $M$ iff the initial configuration of $M$ (state $q_0$ and head at the left of the tape) is accepting, and rejecting if this configuration is rejecting. 
%\end{definition}
%The proof of EXPSPACE-hardness is by reduction from the following EXPSPACE-complete problem:
%\begin{definition}[Acceptance of EXPTIME alternating Turing machine]
%An instance of the problem is an alternating Turing machine whose time is bounded by an exponential, together with a word $w$. 
%The question is if $M$ accepts $w$. 
%\end{definition}
\begin{theorem}
\label{fulllower}
The satisfiability problem for $\fotwo$ is $\expspace$-hard, with hardness holding
even when formulas are restricted to be in $\fotwo[\vdesc]$.
\end{theorem}
This is proved by coding the acceptance problem for an alternating exponential time machine. A tree node
can be associated with an $n$-bit address, either by using multiple predicates (for  $\fotwo[\vdesc]$) or
via children. The equality and successor relations between the addresses associated to nodes $x$ and $y$ can
be coded in $\fotwo$ using the standard argument (see the $\nexptime$-hardness proof of \cite{fo2_utl}).
A path corresponds to one thread of the alternating computation, and the tree structure is used to code
alternation.

% EXPAND AND REMOVE ON ENCODING

\section{Satisfiability without child} \label{sec:nochild}
%ubsection{The sublogic \fod}
%\label{sec:fod}
\myparagraph{The exponential depth bound revisited}
As noted in the previous section, the satisfiability problem
is still $\expspace$-complete even
when the $\childof$ relation is removed. However, we take a closer
look at this case, noting some connections with other logics and
some further restrictions that lower the complexity.

We first consider the relationship of $\fotwo$ without child to modal tree languages.

Let \emph{downward stutter-free $\navxp$}, denoted $\dsfnavxp$, be the
fragment of $\navxp$ obtained by restricting to the descendant,
ancestor, and all sibling axes.  The complexity of satisfiability
$\dsfnavxp$ has not been studied in prior work, including
\cite{BFG}, but we can show the following depth bound for $\dsfnavxp$:
\begin{theorem}
\label{tlsat}
Every satisfiable $\dsfnavxp$ sentence has a model of polynomial
depth.  
%If only the descendant axis is permitted, the branching
%can be taken to be polynomial as well.
The satisfiability problem for $\dsfnavxp$ is $\pspace$-complete.
%(and is $\pspace$-complete).
\end{theorem}

The proof resembles the result that a satisfiable stutter-free
temporal logic formula has a model of polynomial size. Some care needs
to be taken to deal with the sibling axes, which allow a $\dsfnavxp$
formula to look off of a given path.  

This result shows that tight bounds for two-variable logic
without child can actually be obtained via translation to modal languages:
Combining the first part of  Theorem 
\ref{tlsat}  and the translation to $\navxp$
%to $\dsfnavxp$
from Proposition \ref{prop:fo2tl},
we get  an~alternative proof of the exponential depth
bound in Lemma \ref{depthlemma}, as well as the $\expspace$ upper
bound for satisfiability, in the special case of
$\fotwo[\vnochild]$.

\myeat{
The
lower-bound argument can be modified to show that $\fotwo[\vnochild]$
satisfiability is $\expspace$-complete relative to all trees, UAR
trees, and to ranked schemas.
}%myeat

\myparagraph{Unary Alphabet Restriction, polynomial alternation bounds, and polynomial depth bounds} The previous section showed $\expspace$-complete-ness for satisfiability
of $\fotwo[\vdesc]$. However the $\expspace$-hardness argument for $\vdesc$ makes use of multiple predicates
holding at a given node, to code the address of a tape cell of an alternating $\exptime$ Turing Machine.
It thus does not apply to satisfiability over Unary Alphabet Restriction trees (as defined
in Section \ref{sec:prelim})
or to satisfiability with respect to a schema, since schemas restrict to a single alphabet symbol per node.
We show that the complexity of satisfiability is actually ``lower''
(that is, modulo the assumption $\nexptime \neq \expspace$) when the UAR is imposed, using
distinct techniques for the case of ranked and unranked trees.

We start by noting
that one always has at least $\nexptime$-hardness, even with UAR.

\begin{theorem}
\label{desconlyuar}
The satisfiability of $\fotwo[\vdesc]$ with the 
unary alphabet restriction is $\nexptime$-hard, and similarly  with respect to a ranked schema.
\end{theorem}
The proof is a  variation of the argument for $\nexptime$ hardness for words \cite{fo2_utl}, but this time
using the frontier of a shallow but wide tree to code the tiling of an exponential grid.

%\myparagraph{Polynomial alternation bounds}
We will prove a matching $\nexptime$ upper bound for UAR trees and for
satisfiability with respect to a ranked schema.  To do this, we extend
an idea introduced in the thesis of Philipp Weis \cite{WeisPhd},
working in the context of $\fotwo[<]$ on UAR words: polynomial bounds
on the number of times a formula changes its truth value while keeping the same
symbol along a given
path.

The following is a generalization of  Lemma 2.1.10 of Weis \cite{WeisPhd}.

Consider an $\fotwo[\vdesc]$ formula $\psi(x)$, 
a tree $t$ satisfying the UAR, and fix a 
root-to-leaf path $p=p_1 \ldots p_{max(p)}$ in $t$.  Given a label $a$, 
define an \emph{$a$-interval} in $p$ to be a set of the form
$\{ i : m_1 \leq i < m_2 ;\, t,p_i \models a(x) \}$.  
\begin{lemma} \label{lem:weistree}
For every  $\fotwo[\vdesc]$ formula $\psi(x)$,
UAR tree $t$,
and root-to-leaf path $p=p_1 \ldots p_{max(p)}$ in $t$,
the set $\{i| ~ t, p_i \models \psi \wedge a(x)\}$
is made up of at most $|\psi|^2$ $a$-intervals.
\end{lemma}

From Lemma \ref{lem:weistree}, we will show that $\fotwo[\vdesc]$ sentences that
are satisfiable over UAR trees always have polynomial-depth witnesses:
\begin{lemma} \label{lem:smallwitnes} If an $\fotwo[\vdesc]$ formula $\varphi$ is satisfied over a UAR
tree, then it is satisfied by a model of depth bounded by a polynomial
in  $|\varphi|$.
\end{lemma}

\begin{proof}
Suppose that  $\varphi$ is satisfied over a UAR
tree $t$. On each path $p$, for each
letter $b$, let a $b,\varphi$-interval be a maximal $b$-interval on
which every one-variable subformula of $\varphi$ has constant truth value.
By the lemma above, the total number of such  intervals is polynomially bounded.
We let $W$ 
contain the endpoints of each $b, \varphi$-interval for all symbols $b$.
We note the following crucial property of $W$:
for every node $m$ in $p$ which is not in $W$, there  is  a node in $W$ with the
same $\varphi$-type as $m$ that is strictly above  $m$, and also one strictly below $m$.

\begin{figure}[h]
\begin{center}
%\scalebox{0.9}{
\scalebox{0.6}{
\includegraphics{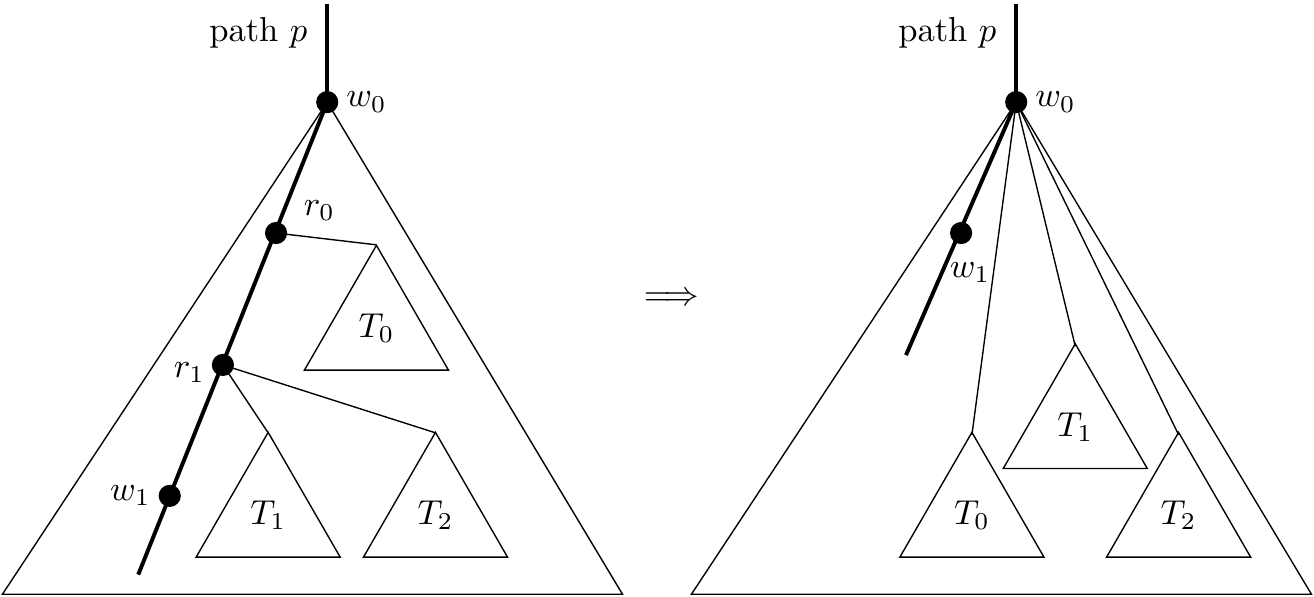}

}
\end{center}

\caption{Tree Promotion}
\label{fig:promote}

\end{figure}

The idea is now to remove all those points on path $p$ that are not in
$W$.  This must be done in a slightly unusual way, by ``promoting''
subtrees that are off the path.  For every removed node $r$, for every
child $c$ of $r$ not on $p$, we attach the subtree rooted at $c$ to
the closest node of $W$ above $r$ (see Figure \ref{fig:promote}).  Let
$t'$ denote the tree obtained as a result of this surgery.  Formally,
the nodes of $t'$ are all nodes of $t$ that are not in $p$ or are in
$W$. Each such node has the same label that it had in $t$.  For any
node $m$ in $t$ with parent $n$, if both $m$ and $n$ are in $t'$ then
$n$ is again the parent of $m$ in $t'$.  On the other hand, if only
$m$ is in $t'$ then its parent in $t'$ is its lowest ancestor in $W$.

%For a node in $W$, its parent in
%$t'$ is its deepest ancestor in $W$.  
%For a node $m$ of $t'$ that is not in
%$p$ and whose parent $m'$ within $t$ also is in $t'$, we let the parent of
%$m$ be $m'$. For a node $m$ of $t'$ whose parent $m'$ is a node of $p$
%that is not in $W$, we set the parent of $m$ to be the deepest
%ancestor of $m'$ in $W$.  See Figure \ref{fig:promote}.

Let $f$ be the partial function taking a node in $t$ that is not
removed to its image in $t'$.  We claim that $t'$ still satisfies
$\varphi$, and more generally that for any subformula $\rho(x)$ of
$\varphi$ and node $m$ of $t$, we have $t,m \models \rho$ iff $t',
f(m) \models \rho$.  This is proved by induction on $\rho$, with the
base cases and the cases for boolean operators being
straightforward. For an existential formula $\exists y \beta(x,y)$, we
give just the ``only if'' direction, which is via case analysis on the
position of a witness node $w$ such that $t,m,w \models \beta$.

If $w$ is in $t'$ then $t',m,w \models \beta$ by the induction
hypothesis and the fact that $w$ is an ancestor (or descendant) of $m$
in $t'$ if and only if it is an ancestor (or descendant) of $m$ in $t$.

If $w$ is not in $t'$, then it must be that $w$ lies on the path $p$
and is not one the protected witnesses in $W$.  But then $w$ has both
an ancestor $w'$ and descendant $w''$ in $W$ that satisfy all the same
one-variable subformulas as $w$ does in $t$, with both $w'$ and $w''$
preserved in the tree $t'$.  If $m$ and $w''$ are distinct then
$t',m,w'' \models \beta$ by the induction hypothesis and the fact that
$m$ and $w''$ have the same ancestor/descendant relationship in $t'$
as do $m$ and $w$ in $t$.  If $m$ is identical to $w''$ then
$t',m,w' \models \beta$ by similar reasoning.  In any case we deduce
that $t',m \models \exists y \beta$.

Since this process reduces both the length of the chosen path $p$ and
does not increase the length of any other path, it is clear that
iterating it yields a tree of polynomial depth.
\end{proof}

Note that we can guess a tree as above in $\nexptime$, and hence
we have the following bound:

\begin{theorem} \label{thm:uarnochild} Satisfiability for $\fotwo[\vdesc]$ 
formulas over
UAR unranked trees is in $\nexptime$, and hence is
$\nexptime$-complete.
\end{theorem}

\myparagraph{Bounds on subtrees and satisfiability of $\fotwo[\vdesc]$ with respect to a ranked schema}
The collapse argument above relied heavily on the fact that trees were unranked, since over a fixed rank
we could not apply ``pathwise collapse''.
Indeed, we can show that over ranked trees, a $\fotwo[\vdesc]$ formula satisfiable
over UAR trees need \emph{not} have a witness
of polynomial depth:

\begin{theorem} \label{rankednochilduar}
There are  $\fotwo[\vdesc]$ formulas $\varphi_n$ of size $O(n)$ that are satisfiable over
UAR binary trees, where the minimum depth of satisfying UAR binary trees grows as $2^{n}$.
\end{theorem}

Nevertheless, we can still obtain an $\nexp$ bound for UAR trees of a~given rank,
and even for satisfiability with respect to a  ranked schema.

\begin{theorem}
\label{fodescup}
The satisfiability problem for $\fotwo[\vdesc]$ over ranked schemas
is in $\nexptime$, and is thus $\nexp$-complete.
\end{theorem}

We give the argument only for satisfiability with respect to rank-$k$ UAR trees, leaving the extension
to schemas for the appendix. This will also serve as an alternative proof of Theorem \ref{thm:uarnochild}.
The idea will be to create a model with only an exponential number of distinct subtrees, which can be represented
by an~exponential-sized DAG.  %We do this by identifying nodes that satisfy similar properties, where
%the nodes may now not be comparable with the descendant relation. 
We do this by creating an equivalence relation that  is globally $\varphi$-preserving (not just pathwise) and which has exponential index (not just path index).
 We will then  collapse  equivalent nodes, as in 
 Lemma~\ref{collapse}. There are several distinctions from that lemma: 
to identify nodes that are not necessarily comparable  we can not afford to  abstract a node by  the set of \emph{all} the types
realized below it, since within the tree as a whole there can be doubly-exponentially many such sets.
Instead
we will make use of  some ``global information'' about the tree, in
the form of a set of ``protected witnesses'', which we denote $W$.
% that can not be victimized by the collapse operation.
%We will use the polynomial alternation bound,
%Lemma \ref{lem:weistree}, to bound the number of witnesses that are protected.

By Lemma \ref{depthlemma} we know that a satisfiable $\fotwo[\vdesc]$ formula $\varphi$ has a model $t$ of depth at most exponential in $\varphi$.
Fix such a $t$.
For each $\varphi$-type $\tau$, let $w_{\tau}$ be a node of $t$ with maximal depth satisfying $\tau$. 
We include all $w_{\tau}$ and all of their ancestors in a set $W$, and call these \emph{basic global witnesses}.
%Consider a node $n$ that is not on the path from root to $w_\tau$ and is not descendant of $w_\tau$. Note that any one variable subformula holding at node $n$ that has an incomparable witness with $\varphi$-type $\tau$ has one already in set $W$. 
For  any   $m$ that is an ancestor or equal to a basic  global witness $w_\tau$, 
and any subformula $\rho(x) =\exists y \beta(x,y)$ of   $\varphi$, if there
is $w'$ incomparable (by the descendant relation) to $m$ such that  $t, m, w' \models \beta$
we add one such $w'$ to $W$, along with all its ancestors -- these are the
\emph{incomparable global witnesses}.
\myeat{
We claim that for every $m \in t$, if there is $w$ incomparable to $m$ with
$t,m, w \models \beta(x,y)$, then there is such a $w$ in $W$.
Fix $m$ and suppose that such a $w$ exists. Let $w_\tau$ be the basic
global witness for the type of $w$. If $w_\tau$ is incomparable to $m$, then
$w_\tau$ has the required property. If $w_\tau$ is a descendant of $m$, then
we would have thrown in the necessary $w$ into $W$ as an
incomparable global witness for $m$.
If $w_\tau$ is an ancestor of $m$ or equal to $m$, we would have thrown in the necessary
$w$ into $W$ as an incomparable global witness for $w_\tau$.
}%moved to appendix

We need one more definition.
Given a node $m$ in a tree, for every $\varphi$-type $\tau$ realized by some ancestor
$m'$ of
 $m$, for every subformula $\exists y \beta(x,y)$ of $\tau$, if
there is a descendant $w$ of $m$ such that $t, m', w \models \beta(x,y)$,
 choose one such  witness $w$ and let $\sdesctype(m)$ include
the $\varphi$-type of that witness. Note that the same  witness will suffice for every ancestor $m'$ realizing $\tau$,
and since there are only polynomial many $\varphi$-types realized on the path,
the collection  $\sdesctype(m)$ will be of polynomial size.

Now we transform $t$ to $t'$ such that $t' \models \varphi$ and
$t'$ has only exponentially many different subtrees.  We make use of a
well-founded linear order $\prec$ on trees with a given rank and label
alphabet, such that: \begin{inparaenum} \item
$\subtree(t,n') \prec \subtree(t,n)$ implies $n'$ is not an ancestor
of $n$; \item for every tree $C$ with a distinguished leaf, for tree
$t_1, t_2$ with $t_1 \prec t_2$, we have $C[t_1] \prec C[t_2]$, where
$C[t_i]$ is the tree obtained by replacing the distinguished leaf of
$C$ with $t_i$.
\end{inparaenum}
There are many such orderings, e.g. using standard string encodings of a tree.

For any model $t$ 
if there are two nodes $n, n'$ in $t$  such that
\begin{inparaenum}
\item $n, n' \not \in W$, 
%\item the subtree rooted at $n$ is not equal to the subtree rooted at $n'$
\item $\tp_\varphi(n)=\tp_\varphi(n')$,
\item $\anctype(n) = \anctype(n')$,
\item $\sdesctype(n)$ $=$ $\sdesctype(n')$,
\item $\subtree(t,n') \prec \subtree(t,n)$ (which implies that  $n'$ cannot be an ancestor of $n$),
\end{inparaenum}
then let $t'=\up(t)$ be obtained
by choosing such $n$ and $n'$
and  replacing the subtree rooted
at  $n$ by the subtree rooted at $n'$.
%Note that $t'$ also contains a copy of $W$.

Let $T_1$ be
the nodes in $t$ that were not in $\subtree(t,n)$,
and for any node $m \in T_1$ let $f(m)$
denote the same node considered within $t'$.
Let
$T_2$ denote the nodes in $t'$ that are images
of a node in $\subtree(t,n')$.
For each $m \in T_2$,
let $f^{-1}(m)$ denote the node in $\subtree(t,n')$ from which  it derives.

We claim the following:
\begin{lemma} \label{lem:collapsevdesc}
For all $m \in T_1$ the 
$\varphi$-type of  $n$ in $t$ is the
same as the $\varphi$-type
of $f(m)$ in $t'$.
Moreover, for every node $m'$ in $T_2$, 
the 
$\varphi$-type of $m'$ in $t'$ is
the same
as 
that of  $f^{-1}(m)$ in $t$.
\end{lemma}

Applying the lemma above to the root of
$t$, which is necessarily in $T_1$, it  follows that the truth
of the sentence $\varphi$ is preserved by this operation.

We now iterate the procedure $t_{i+1}~ := ~ \up(t_i)$, until
no more updates are possible. This procedure terminates, because the tree decreases in the order $\prec$ every step. 
We can thus represent the tree as an exponential-sized DAG, with one node for each subtree.

Thus we have shown that any satisfiable formula has an exponential-size
DAG that unfolds into a model of the formula. Given
such a DAG, we can check whether an $\fotwo$  formula holds
in polynomial time in the size of the DAG.
This gives
a $\nexptime$ algorithm for checking satisfiability.

\section{Satisfiability without descendant} \label{sec:nodesc}

%We know that this is NEXPTIME hard and in EXPSPACE.

%\subsection{$\fotwo[\vchild]$}

Recall that even on words with only
the successor relation, the satisfiability problem for two-variable logic  is
$\nexptime$-hard \cite{fo2_utl}. From this it is easy to see that the satisfiability for
$\fotwo[\vchild]$ is $\nexptime$-hard, on ranked and unranked trees.

%However, Weis \cite{WeisPhd} has shown that under the unary alphabet
%restriction (every node has only one proposition satisfied on it), that
%a satisfiable $\fotwo$ sentence has a model of polynomial size, and hence
%the satisfiability problem for two-variable logic on words  in in NP.
%We may expect  that the unary alphabet restriction can have
%an impact on satisfiability of $\fotwo[\vchild]$ as well.

%We first show that the satisfiability problem for $\fotwo[\vchild]$
%is $\nexptime$-hard.

\begin{theorem}
\label{fod}
The satisfiability problem for $\fotwo[\vchild]$ is $\nexptime$-hard,
even with the unary alphabet restriction.
\end{theorem}

We now present a matching upper bound, which holds even in the presence of sibling relations, i.e., for $\fotwo[\vnodesc]$.
The result is surprising, in that it is easy
to write satisfiable $\fotwo[\vchild]$ sentences $\varphi_n$ of polynomial
size whose 
smallest tree model is of depth exponential in $n$, and whose size is
doubly exponential. Indeed, such formulas can be obtained
as a variation of the proof of Theorem \ref{fod}, by coding a complete binary
tree whose nodes are associated with $n$-bit numbers, increasing
the number by $1$ as we move from parent to either child.

The result below relies on the fact that one can witness the
satisfiability of a~given formula by an exponential-sized DAG.

\begin{theorem}
\label{fodup}
The satisfiability problem for $\fotwo[\vnodesc]$, and the satisfiability problem with respect
to a rank schema,
are in
$\nexptime$, and hence are $\nexptime$-complete.
\end{theorem}

We sketch the idea for satisfiability, which iteratively quotients the structure by an equivalence relation, while
preserving certain global witnesses, along the lines of Theorem \ref{fodescup}.
By Lemma \ref{depthlemma}
 we know that a
satisfiable $\fotwo[\vnodesc]$ formula~$\varphi$ has a model $t$ of 
depth at most exponential in $\varphi$, where
the outdegree of nodes is bounded by an exponential.

For each $\varphi$-type that is satisfied in $t$,
 choose a witness and include 
it along with all its ancestors in a set $W$ -- that is, we include the ``basic witnesses'' as
in Theorem \ref{fodescup}. We also include all children of each basic witness -- call these ``child witnesses''.

Thus the size of the set of ``protected witnesses'' $W$ is again at
most exponential.  Now we transform $t$ to $t'$ such that
$t' \models \varphi$ and at the same time $t'$ has only exponentially
many different subtrees.  Our update procedure looks for nodes $n, n'$
in $t$ such that
\begin{inparaenum}
\item $n,n' \not \in W$;
\item $\subtree(t,n') \prec \subtree(t,n)$, where $\prec$ is an appropriate ordering (as in Theorem \ref{fodescup});
\item $\tp_{\varphi}(n)=\tp_{\varphi}(n')$  and  $\tp_{\varphi}(\text{parent}(n))=\tp_{\varphi}(\text{parent}(n'))$.
%the parent of $n$ has the same $\varphi$-type as the parent of $n'$
\end{inparaenum}
We then obtain $t'=\up(t)$ 
by choosing such $n$ and $n'$
and  replacing 
$\subtree(t,n)$ by $\subtree(t,n')$.
%Note that $t'$ also contains a copy of $W$.

The theorem is proved by showing that this update operation
preserves~$\varphi$. Iterating it until no two nodes can be found
produces a tree that can be represented as an exponential-size DAG.

%\section{$\fotwo$ on probabilistic schemas} \label{sec:probxml}
%\input{probxml}

%\section{Related Work} \label{sec:related}
%To mention: EVW, Immerman et. al., Boxes and Diamonds paper

\section{Conclusions} \label{sec:conc}
We have shown that the parallel between the complexity of $\fotwo$ satisfiability on general structures and
on restricted structures breaks down as we move from words to trees -- trees allow one to encode alternating
exponential time computation, leading to $\expspace$-hardness. On the other hand, we show that analogs
of the ``model shrinking'' methods for $\fotwo$ on words exist for trees, albeit using a different
shrinking
technique. In future work, we are extending the analysis to infinite trees, where we believe it can
be useful for analyzing branching time properties  of both non-deterministic and probabilistic systems,
as was done for linear time in \cite{BLW2}. We are also considering the case of
structures of fixed tree-width.

Our main complexity results on satisfiability
are summarized in Table \ref{thetable}, where in each case the  bound is tight.

\begin{center}
\begin{tabular}{|l|c| c|c|c|}
\hline
                                & $\fotwo$ & $\fotwo[\vdesc]$  & $\fotwo[\vnochild]$     & $\fotwo[\vchild]$ \\
\hline

All Trees          & $\expspace$ & $\expspace$  & $\expspace$ & $\nexptime$ \\
w.r.t. Ranked Schema         &  $\expspace$ & $\nexptime$  & $\expspace$ & $\nexptime$ \\
%Depth of witness  & $2^{2^{poly(|\varphi|)}}$ & $2^{poly(|\varphi|)}$ & $2^{poly(|\varphi|)}$  \\
\hline
\end{tabular} \label{Complexity of satisfiability}
\label{thetable}
\end{center}

\bibliographystyle{alpha}
\bibliography{litb}

\newcommand{\noopsort}[1]{} \newcommand{\printfirst}[2]{#1}
  \newcommand{\singleletter}[1]{#1} \newcommand{\switchargs}[2]{#2#1}
\begin{thebibliography}{BMSS09}

\bibitem[BFG08]{BFG}
Michael Benedikt, Wenfei Fan, and Floris Geerts.
\newblock {XPath} satisfiability in the presence of {DTDs}.
\newblock {\em J. ACM}, 55(2):8:1--8:79, 2008.

\bibitem[BK09]{leashed}
Michael Benedikt and Christoph Koch.
\newblock {XPath Leashed}.
\newblock {\em ACM Comput. Surv.}, 41(1), 2009.

\bibitem[BLW12]{BLW2}
Michael Benedikt, Rastislav Lenhardt, and James Worrell.
\newblock Verification of two variable logic revisited.
\newblock In {\em QEST}, 2012.

\bibitem[BMSS09]{datatrees1}
Mikolaj Bojanczyk, Anca Muscholl, Thomas Schwentick, and Luc Segoufin.
\newblock {Two-variable logic on data trees and XML reasoning}.
\newblock {\em J. ACM}, 56(3), 2009.

\bibitem[Boa97]{tilepaper}
Peter Van~Emde Boas.
\newblock The convenience of tilings.
\newblock In {\em In Complexity, Logic, and Recursion Theory}, 1997.

\bibitem[EVW02]{fo2_utl}
K.~Etessami, M.~Y. Vardi, and T.~Wilke.
\newblock First-order logic with two variables and unary temporal logic.
\newblock {\em Inf. and Comp.}, 179(2):279--295, 2002.

\bibitem[Fig12]{twosuccs}
Diego Figueira.
\newblock Satisfiability for two-variable logic with two successor relations on
  finite linear orders.
\newblock {\em CoRR}, abs/1204.2495, 2012.

\bibitem[GKV97]{kolaitisvardi}
Erich Gr\"adel, Phokion~G. Kolaitis, and Moshe~Y. Vardi.
\newblock On the decision problem for two-variable first-order logic.
\newblock {\em Bulletin of Symbolic Logic}, 1997.

\bibitem[LS08]{xpathutl}
Leonid Libkin and Cristina Sirangelo.
\newblock Reasoning about xml with temporal logics and automata.
\newblock In {\em LPAR}, 2008.

\bibitem[Mar04]{marx04}
Maarten Marx.
\newblock {XPath} with conditional axis relations.
\newblock In {\em EDBT}, 2004.

\bibitem[MdR04]{derijkemarx}
Maarten Marx and Maarten de~Rijke.
\newblock {``Semantic Characterizations of XPath''}.
\newblock In {\em TDM}, 2004.

\bibitem[Tho97]{thomashandbook3}
Wolfgang Thomas.
\newblock {``Languages, Automata, and Logic''}.
\newblock In G.~Rozenberg and A.~Salomaa, editors, {\em Handbook of Formal
  Languages}. Springer, 1997.

\bibitem[Wei11]{WeisPhd}
Philipp Weis.
\newblock {\em Expressiveness and Succinctness of First-Order Logic on Finite
  Words}.
\newblock PhD thesis, University of Massachusetts, 2011.

\end{thebibliography}
\newpage
%\appendix
\section*{More detail on the proof  of Lemma \ref{depthlemma} and Lemma \ref{collapse}}

We first give a detailed proof of the following statement from Lemma \ref{collapse}:

\emph{The equivalence relation $\equivfull$ is pathwise $\varphi$-preserving.}

Fix tree $t$
and  $n_0 \equivfull n_1$ lying on the same path $p$ in $t$, with $n_1$ a descendant of $n_0$.  Let $t'$ be formed
by overwriting $n_0$ with $n_1$, and 
$f$ be the mapping taking a node that lies in the subtree of  $n_1$ or outside
of the subtree of $n_0$ to its image in $t'$. By the ``collapsed part of $t$'' we refer to the part of $t$ not in the domain
of $f$. 

We prove  via structural induction that for every subformula $\rho$ of $\varphi$
and node $m$ in the domain of $f$ we have $t,m \models \rho \leftrightarrow t',f(m) \models \rho$.
The atomic cases and the boolean operators are clear, so existential quantification is
the only non-trivial case.

Consider  first a node $m$ in the bottom half of the non-collapsed
structure -- that is, in $\subtree(t,n_1)$  -- satisfying
 $\rho(x)=\exists y \beta(x,y)$.
By induction
we need consider only the case where
some node $w$ witnessing
that $m$ satisfies $\rho$ in $t$
is not in the domain of $f$. Fix such a witness node $w$.
We  show that we can find a node that satisfies the same one-variable subformulas
of $\rho$ that $w$ does, and which satisfies the same axis relations with respect to $m$ that
$w$ does. 

When the witness to the existential quantifier in $\rho$ is a parent of $m$, then we must have $m=n_1$.
Now we can apply the hypothesis that the $\varphi$-type of  $n_0$ is the same as the $\varphi$-type
of $n_1$, plus the induction hypothesis, to conclude that $f(m)$ must satisfy $\rho$.
The case in which the
witness $w$ is
a descendant  of $m$ or equal to $m$ need not be considered,
since such a witness must be in the domain of $f$, which is ruled out by assumption.
Now consider the case where some witness $w$ is an ancestor of $m$, but not a parent. Such a $w$ must be on the path $p$. 
In this case, we can use the fact that $\anctype(n_0)=\anctype(n_1)$ to argue that a~witness can be found.
Suppose there is a node $w$ witnessing that $t,m \models \rho(x)$
such that $w$ is  not an ancestor or a descendant of $m$.   Then we can apply
the fact that $\contexttype(n_0)=\contexttype(n_1)$ to find a witness $w'$ that is incomparable
of $n_0$, but still in the domain of $f$. Such a $w$ can be used (by induction) as a~witness that
$t', f(m) \models \rho$.

We now move to the case where  $m$ is in the top half of the non-collapsed
structure satisfying
 $\rho(x)=\exists y \beta(x,y)$. We are interested in the case
where all witnesses $w$ to the existential quantifier in $\rho$ are in the collapsed part of the
structure, and hence are not ancestors of $m$.

Suppose we have a witness that is not a descendant or ancestor of $m$. The witness
must be 
a descendant  of $n_0$, and $n_0$ must not be a descendant of $m$.  
We can apply again the fact that $\desctype(n_0)=\desctype(n_1)$ to find a witness $w'$ below
$n_1$, which will suffice by induction.
%there is a non-descendant $w' \in W$ satisfying the same one-variable subformulas, and which is  in
%the lower part of the pre-collapse structure structure. Such a $w'$ is then not a descendant of $m$, since
%$m$ is not on $p$. Hence $w'$ can be used as a witness.

If the witness $w$ is in the collapsed part of $t$ and is a child of $m$, we must have $m=n_0$, and
hence we can use the fact that $\tp_\varphi(n_0)=\tp_\varphi(n_1)$ to get the desired witnessed.
Now suppose we have a witness $w$ in the collapsed part
of the structure, with $w$ a descendant of $m$ but not a child of $m$. Again, if $m=n_0$ we are done,
using the fact
that $\tp_\varphi(n_0)=\tp_\varphi(n_1)$.  If $m \neq n_0$, we must
have $m$ is a strict ancestor of $n_0$.
From $\desctype(n_0)=\desctype(n_1)$ we know that there is a descendant $w'$ of $n_1$ with the same $\varphi$-type
as $w$. Since $m \neq n_0$ $w'$ is not a child of $m$ in $t'$, and hence can serve as a witness.

The cases for the sibling axes are also straightforward, since no nodes in the domain of $f$ have their siblings
modified by the collapse mapping.

~
~

We now explain the variation of the argument for  the exponential bound on branching.
Note that $\navxp$ queries can already force exponential branching, and thus the result
does not follow directly via translation to modal tree logics.
In a nutshell, we use the same approach, but shrinking horizontal rather than vertical
paths.

{\bf Construction:} Consider  the equivalence relation that relates two nodes  if they have: 
\begin{compactitem}
\item the same $\varphi$-types that occur as left-siblings, and the same $\varphi$-types
that occur as right-siblings
\item  the same $\varphi$-types of nodes that are descendants of right-siblings, and similarly for left-siblings
\item the same $\varphi$-types, and the same $\varphi$-types immediately to the right and immediately to the left
\end{compactitem}
Recall that  the right-sibling relation is the transitive closure of the immediate right-sibling relation, and similarly
for left-sibling.
Note that the first two items change only exponentially many times, and on an interval
where they are both constant, the third item takes on only exponentially many values.

We now claim that any sufficiently long horizontal path can be pruned.
Fix a~horizontal path $p$ containing all children of some node. If $p$
is sufficiently long, there is some equivalence class $C$ that
has more than one node in it.
Let $n'$ be the left-most (lowest in sibling order) element of $C$, and $n$ the element
of $C$ that is closest to it on the right.
Let $t'$ be obtained by removing all subtrees of nodes between $n'$ and $n$, including
the subtree of $n$ but not the subtree of $n'$.

{\bf Correctness:}
Let $f$ be the function taking a node in $t$ that was not removed by the operation
above (for short ``non-removed node'') to its image in $t'$.
As usual, we proceed by showing that $\varphi$-types are preserved in moving from a~node $m$ to $f(m)$.
As before, the only important case is the inductive step for $\rho(x)=\exists y \beta(x,y)$, with
the non-trivial direction being to show that if $\rho$ holds at $t,m$ then it holds in $t', f(m)$.
Suppose $m$ satisfies $\rho$, with witness $w$. 
The interesting case is when $w$ is a removed node, 
which means it must either
be a right-sibling of $n'$ that was removed or
 below a right-sibling of $n'$ that was removed.
We do case analysis on the relationship of $w$ to $m$.

{\bf Case of Incomparable Witnesses:}
If $w$  is incomparable to $m$ by both the sibling and ancestor relations, then we consider
several subcases. 

The first subcase is where  $m$ is ``below a node in $p$'' -- that is, a descendant of some node
on $p$. Let  $n''$ be the node of $p$ that is an ancestor of $m$.

We further consider the  subsubcase where the sibling $n''$ is to the right of $n$.
If $w$ is a right-sibling of $n'$, then it was a left-sibling of $n$ or is equal to $n$, since these are the 
siblings that are removed. In the first case,
 it must be that $n'$ has a left-sibling $w'$ with the
same $\varphi$-type as $w$. Since $m$ is ``down and to the right'' (that is, below a right-sibling) of $n'$, 
$w'$ is incomparable to $m$, and thus
such a~$w'$  can be used as a witness that
$t', f(m) \models \rho$. Similarly, in the case that $w$ was equal to $n$, $n'$ can be used as a witness.
If $w$ is below a right-sibling of $n'$, it must be that $n'$ has a left-sibling that has
a descendant with the
same $\varphi$-type, and this can be used as a witness.

The paragraph above completes the subsubcase where $n''$ is to the right of $n$.
If $n''$ is to the left of, or is equal to, $n'$, we argue symmetrically,
but considering the $\varphi$-types that are right-siblings or descendants of right-siblings
 of $n$.

The subcase where $m$ is itself a sibling of $n$ is similar to the above, except
$w$ can not be a sibling of $m$, and hence one subcase does not need to be considered.

The final subcase is where   $m$ is not  on $p$  and is not a descendant of a node in $p$.
Note the assumption that  $w$ is incomparable to $m$ and  removed during the collapse process,
and hence $w$ lies
below a node on the horizontal path $p$. This  implies
that $m$ can not be an ancestor of the nodes in $p$.
If  $w$ is
a sibling of a~node in $p$ that was removed, we can use any non-removed sibling of $n'$ with
the same $\varphi$-type as a witness (there are at least two such nodes, to the left and right). Similarly if $w$ is
 below a sibling of  a removed node of $p$, we use any non-removed node that has the $\varphi$-type of $w$
and which is a descendant of a node on $p$.

{\bf Other cases:}
The case where the witness $w$ is a descendant of $m$ is similar to the last subcase above. In this
case, $m$ must be an ancestor of the nodes on $p$. Again, if $w$ is a sibling of $n$,
we can choose a sibling with the same $\varphi$-type. If $w$ is a descendant of a sibling,
we can choose a descendant of a sibling with the same $\varphi$-type.

We now turn to the case where $w$ is an immediate left-sibling of $m$. In this case
we must have $m=n$, and we can use the fact that $n$ and $n'$ have the same 
$\varphi$-type for their immediate left-siblings. The case where $w$ is an immediate right-sibling of $m$ is analogous.

The case where $w$ is a following-sibling but not the next-sibling, or a preceding-sibling
but not the previous-sibling, is handled similarly to  above.

Iterating this pruning process gives the required branching bound.

\section*{Proof of Theorem \protect{\ref{navxpbound}}}

Recall the statement:

\emph{The satisfiability of a $\navxp$ filter $\varphi$ over trees
of bounded depth~$b$ is in $\pspace$ (in $b$ and $|\varphi|$).
}

%Since it is more convenient to deal with one-variable formula than a mix of two- and one-variable
%as in $\navxp$, we
%will prove this for the  modal tree logic $\tl$.
It can be awkward to work with $\navxp$, since one has to switch back between two- and one- variable
formulae. For simplicity, we work with 
 a temporal logic $\utl$ for trees analogous to
Unary  Temporal Logic on words,
%\myeat{
introduced in \cite{xpathutl}.
Formulas $\varphi$ are given by:
\[ \varphi \; ::= \; P_i \;|\; \varphi \wedge \varphi \;|\;  \neg \varphi \;|\;
 \LTLdiamond_* \varphi \;|\; \LTLdiamondminus_* \varphi
\;|\; \LTLcircle_* \varphi \;|\; \LTLcircleminus_* \varphi \]
where $*$ stands for either a child ($\ch$) relation or a next-sibling relation ($\ns$).
Informally
$\LTLdiamond_\ch \varphi$ is ``eventually along a vertical path $\varphi$ holds'', $\LTLdiamondminus_\ch$ is ``up the vertical path to the root'', 
$\LTLcircle_{\ch}$ is ``in some child'' and
$\LTLcircleminus_\ch$ ``in the parent''. The variants for $\ns$ are defined similarly for horizontal paths.
The semantics of $\utl$ with respect to a tree $T$ and node $s$ is given as
a variant of the standard semantics for linear temporal logic on words.
For example $(T, s) \models P_i \iff s \mbox{ has label } \mathbf{P_i}$. The boolean operators
have their usual recursive definition.
%\begin{align*}
%(T, s) &\models P_i \iff s \in \mathbf{P_i} \\
%(T, s) &\models \varphi_1 \wedge \varphi_2 \iff (T, s) \models \varphi_1 \text { and } (T, s) \models \varphi_2 \\
%(T, s) &\models \varphi_1 \vee \varphi_2 \iff (T, s) \models \varphi_1 \text { or } (T, s) \models \varphi_2 \\
%(T, s) &\models \neg\varphi \iff (T, s) \not\models \varphi \\
$(T, s) \models \LTLcircle_{\ch} \varphi \iff \exists s' \text{ such that } s' ~ \childof ~ s \text {  and } (T, s') \models \varphi$, and similarly for the other next state modalities.
%We have that $\LTLdiamond_{\ch} \varphi \equiv \top \U_{\ch} \varphi$ (`eventually along some path $\varphi$ holds') and
%similarly $\LTLdiamond_{\ns} \varphi \equiv \top \U_{\ns} \varphi$ (`eventually along right sibling $\varphi$ holds').
%$\LTLdiamondminus_{\ch} \varphi \equiv \top \S_{\ch} \varphi $ (`previously at some grandparent node $\varphi$ holds') and
%similarly $\LTLdiamondminus_{\ns} \varphi \equiv \top \S_{\ns} \varphi $ (`previously at some left sibling node $\varphi$ holds').

The above semantics maps a formula to a set of nodes in a tree.
For a tree $t$, we say
$t \models \varphi$ to mean $(t, n_0) \models \varphi$ where $n_0$ is the root.

\cite{xpathutl} shows that $\navxp$  can be translated in polynomial time into $\utl$.

We give a non-deterministic $\pspace$ algorithm that constructs a witness tree for $\varphi$, materializing
only the rightmost branch of the tree. As an abstraction of this branch the algorithm guesses
all the $\varphi$-types of nodes appearing on the path to the root, along with auxiliary information
about whether a node is the last child of its parent, and which  subformulas of the
form $\LTLcircle_{\ch} \psi$ and $\LTLdiamond_{\ch} \psi$ have been satisfied.

We require all guessed types to be internally consistent, and to satisfy certain consistency properties.
Additionally, we require $\varphi$ to be in the type of the root.

Now we show how to check the consistency for all temporal subformulas. 
\begin{enumerate}
\item Subformulas $\LTLcircleminus_{\ch} \psi$ and $\LTLdiamondminus_{\ch} \psi$ are the easiest to check, because for each node we 
have already guessed all its ancestors.
\item When we extend a path downward (corresponding to guessing the type of the initial child),
we require that all subformulas $\LTLcircleminus_{\ns} \psi$ are false and that the truth value 
of $\LTLdiamondminus_{\ns} \psi$ is equivalent to truth value of $\psi$. 
When  we move from a leaf $l$ of a path
 to its sibling, we enforce that the new type contains $\LTLdiamondminus_{\ns} \psi$ if $l$ contains it,
and that it contains $\LTLcircleminus_{\ns} \psi$ iff $l$ contains $\psi$.
\item When we move to a sibling of $l$, 
if $l$ contains $\LTLcircle_{\ns} \psi$, we ensure that the type of the newly-created sibling contains $\psi$.
For  $\LTLdiamond_{\ns} \psi$, we guess that its sibling contains $\psi$ or $\LTLdiamond_{\ns} \psi$.
If we guess that a leaf is the rightmost sibling, we check that its type does not contain $\LTLdiamond_{\ns} \psi$.
\item For subformulas $\LTLcircle_{\ch} \psi$ and $\LTLdiamond_{\ch} \psi$, we mark whether they have already been satisfied by
some prior descendant. If not, we decide when we extend the path whether or not they will be satisfied on the new child, and
guess the type accordingly. When we move from a leaf $l$ to its sibling, we require that every such formula that was in $l$ has
been marked as satisfied.
\end{enumerate}

\section*{Proof of Theorem \protect{\ref{fulllower}}}

Recall the statement:

\emph{The satisfiability problem for $\fotwo$ is $\expspace$-hard, and the
same holds for $\fotwo[\vdesc]$.}

We first give the argument for $\fotwo$.  We reduce from the problem
of determining whether an alternating $\exptime$ Turing Machine $T$
accepts a given input $I$.  Without loss of generality we assume that
each configuration of $T$ has exactly two successors.  We can also
assume that for an input of size $n$, the computation of $T$ takes at
most $2^n$ steps and therefore uses at most $2^n$ tape cells.  We give
a polynomial time transformation that takes $T$ and machine input $I$,
returning an $\fotwo$ formula $\varphi$ which is satisfiable if and
only if $T$ accepts $I$.

We encode each tape configuration as a sequence of $2^n$ nodes with one node per cell. Each cell will have a label encoding:
\begin{itemize}
\item the tape symbol written on the cell
\item the time step (or ``index'') of the configuration, encoded in $n$ bits $c_1, c_2, \ldots c_n$
\item the cell position encoded in $n$ bits $p_{1}, p_{2}, \ldots p_{n}$
\item the control state of the Turing Machine
\item the last alternation choice, which is either $\wedge$ or $\vee$ %$\andstate$ or $\orstate$
\item whether the head of the Turing Machine is present
\end{itemize}

The computation of $T$ will be described by a tree of tape
computations starting with an initial configuration.  Intuitively the
formula $\varphi$ will force the shape of the tree to match that of
the computation tree for $T$.  In more detail, an
$\wedge$-configuration will be represented in the tree by a path of
$2^n$ nodes that terminates in a node with two children, each of which
is the root of a successor configuration.  On the other hand an
$\vee$-configuration is represented by a path of $2^n$ nodes that
terminates in a node with a single child, which is the root of a
single successor configuration.  The vocabulary of the formula will
have predicates for the presence or absence of the Turing machine
head, the alternation choice, the tape alphabet symbols, and
predicates indicating which of $c_1, \ldots c_{n}, p_1, \ldots p_n$
hold.

Now we discuss in more detail the parts of $\varphi$ that will ensure the structure described above. 
The
tree should  have as root a node whose index is a vector of zeros for the values of $c_1, \ldots c_n, p_1, \ldots p_n$,
 after which we need to increase the number represented by this vector by one for each child node. Within the same configuration
the latter can be easily enforced by the following formula:

$$ \forall x \, \forall y\, (y \mathrel{\childof} x) \rightarrow
\bigvee_i ( \neg p_i(x) \wedge p_i(y) \bigwedge_{j<i} p_j(x) \leftrightarrow p_j(y) \wedge \bigwedge_{j>i} p_j(x) \wedge \neg p_j(y)
)
$$

We can use the predicates $c_i$ and $p_i$ (and formulas similar to the one above) to determine whether two
nodes  $x$ and $y$ corresponding to tape cells in a configuration of $T$ represent the same, previous or next position within
the same configuration,  or whether they are in the same, previous, or next configuration.
For example, two nodes that represent successive configurations in a single thread of a  machine will
need to be in the $\descof$ relation, and will have configuration co-ordinates that are in a successor relation, which
will be enforced as above, but using the $c_i$ rather than the $p_i$.
%The only additional thing we need to check is whether positions $x$ and $y$ are in a descendant relation.

To encode the alternation, we need to enforce
that the shape of a node is consistent with the type of  the current configuration, in terms of whether
the  state is universal or existential. 
For example, if we have a universal state $q$ and a transition to control states $q_1$ and $q_2$,
after the last cell of the configuration we will enforce that there is a child whose control state is $q_1$
and another child whose control state is $q_2$.
%want their to be two children, each rooting
%a new configuration -- this is easily stated using the $\childof$ relation.
%both branches of computation to be present:

%$$\forall x \, \andstate(x) \wedge \bigwedge_{i =n+1}^{2n} c_i) \rightarrow
%(\exists y \, (y \mathrel{\childof} x \wedge \orstate(y))
%\wedge \exists y \, (y \mathrel{\childof} x \wedge \neg \orstate(y)))
%$$

We have a formula $\psi(x,y)$ that checks the consistency of the tape cells represented by nodes $x$ and $y$ that
are in a descendant relationship (and hence represent the same thread in the alternating computation).
If $x$ and $y$ point to the same cell position in  consecutive configurations then we need the content of $x$, $y$ and their adjacent cells to be consistent with the transition function of $T$, the position of the head, the current state, the cell symbols and the alternation type 
($\wedge$ vs $\vee$).  

The enforcement that the input is on the tape initially, and that an acceptance state is reached at each leaf, can similarly be
easily enforced.

%Note that there are only polynomially many possibilities for these contents (as they include all the information from the label except the indexes of cell and configuration), so $\psi$ can be constructed in polynomial time.

\myparagraph{Extension of the argument from $\fotwo$ to $\fotwo[\vchild]$}
In the proof above we  use only the $\descof$ and $\childof$ relations. 
We now show  how to avoid $\childof$. The key is that we do not need consecutive positions within the same configuration
to occur in a parent child relationship. Along any thread, we can
uniquely identify via the predicates $c_1 \ldots c_n$ and $p_1 \ldots p_n$. 
We can thus consider nodes  correspond to consecutive positions in the same configuration 
using these predicates, while using $\descof$ to restrict to nodes within the same thread.
% $c_1, \ldots c_n$, then we could express that $y$ is child of $x$ just by checking whether their indices correspond to consecutive numbers 
%and at the same time $y \mathrel{\descof} x$.

%Above, we enforce consecutive indices by a formula saying that 
%the index of a child is bigger by one than the index of its parent. 
We will enforce that
\begin{itemize}
\item each descendant of any node has a larger  configuration index
\item each node (except the first) has an ancestor whose configuration address is smaller by one
\item each  node is either a representative of the last configuration in its thread (i.e. with maximal configuration index)
or it has a descendant whose configuration index is higher by one
%\textbf{NOTE: this is not really needed if we don't require full tree}
\end{itemize}
We have similar requirements for the position indices for the same configuration.

\section*{Proof of Theorem \ref{tlsat}}

Recall the key statement:

\emph{Every satisfiable $\dsfnavxp$ sentence has a model of polynomial depth.}
%If only the descendant axis is present, there is a polynomial bound on the branching.

Again, since it is more convenient to deal with one-variable formula than a mix of two- and one-variable
as in $\navxp$, we 
will prove this for the  modal tree logic formed from $\utl$ by
removing the child and parent modalities (but including
the next- and previous- sibling modalities). Call the resulting
language $\tl$.

Consider a satisfiable $\tl$ formula $\varphi$, a tree $t$ satisfying $\varphi$, and a path $p$ in $t$.
We will shrink $p$ to polynomial size without impacting $\varphi$, and iterating this process we can 
achieve polynomial depth.
Once we achieve polynomial depth, we can use Theorem \ref{navxpbound} to get a $\pspace$ bound.

The \emph{vertical $\varphi$-type} of $n$ is defined as the collection of subformulas of $\varphi$ of the form $\LTLdiamond_{\ch}  \psi$ or
$\LTLdiamondminus_{\ch} \psi$ that hold at $n$, along with the formula $a(x)$ where
$a$ is the label of $n$.

The following lemma generalizes an obvious fact about the usual stutter-free temporal logic on words:
\begin{lemma} \label{lem1}
There are polynomially many (in $|\varphi|$) vertical $\varphi$ type changes along any path $p$.
\end{lemma}
\begin{proof}
Consider a path $p$ of $T$ and a node $n$ of $p$. If $n  \not\models \LTLdiamond_{\ch}  \psi$, then in all subsequent nodes $n'$ in the path, 
$n' \not\models \LTLdiamond_{\ch} \psi$. Similarly if $n  \not\models \LTLdiamondminus_{\ch} \psi$, then in all previous nodes $n'$ in the path, $n' \not\models \LTLdiamondminus_{\ch} \psi$.
We therefore have that these subformulas change their truth assignment at most once in $p$.
\end{proof}

We are now ready to prove the polynomial depth bound.
Consider any (downward) path $p$ in the tree. 
By Lemma \ref{lem1}, there are polynomially many vertical type changes along a path.

Consider a maximal interval of $p$ all of whose nodes have the same vertical type, and let $n_{\topn}$ and $n_{\bottomn}$
be the first (highest) and last (lowest) nodes of the interval.
Now consider the tree $t'=t(n_{\bottomn} \rightarrow n_{\topn})$ constructed by overwriting
 $n_{\topn}$ with $n_{\bottomn}$.
%removing the subtree of $n_{\top}$ (not including $n_{\top}$) and attaching to $n_{\top}$ all the 
%children of the parent of $n_{\bottom}$ (including $n_{\bottom}$, and its sibling subtrees),
%with the same sibling ordering.
%Note that this construction differs from the promotion construction used in 
%Lemma \ref{lem:smallwitnes}.

Let $f$ be
the partial function taking nodes in $t$ that are not removed to their images in $t'$.

As with all of our collapse operations, our goal is to show:
\begin{claim} For any subformula $\rho$ of $\varphi$ and node $m$ in the domain of $f$, we have that 
$t, m \models  \rho \leftrightarrow t', f(m) \models \rho$.
\end{claim}

Thus performing this operation on every interval shrinks $p$ without impacting $\varphi$, and iterating
over all $p$ gives the depth bound. 
We prove this by induction on  $\rho$.
Atomic propositions and boolean combinations are immediate.

We begin by considering $\rho = \LTLdiamond_{\ch} \psi$.
If $t, m \models \rho$ then there is a node $w$ below $m$ satisfying $\phi$ in $t$. 
If $w$ is in the domain of $f$, we are done
by induction, so assume $w$ is a descendant of $m$ that is not in the domain of $f$. Thus
$w$ is also a descendant of $n_{\topn}$. Since $n_{\topn}$ has the same downward-type as $n_{\bottomn}$,
$n_{\bottomn}$ has a descendant satisfying $\rho$, and this can be used as a witness.
In the other direction, assume $t', f(m) \models \rho$. There must therefore be a path of nodes in $t$ starting with $m$ leading
to a node $w'$ where $\psi$ holds, and
$w'$ must be of the form $f(w)$ for $w$ in $t$. By induction
$w$ can be used as a witness that $t, m \models \rho$.
A similar argument  holds for $\rho = \LTLdiamondminus_{\ch} \psi$. 

Note that the sibling nodes of a given node $m$ in the domain of $f$ are not impacted by the overwrite
operation.  Using
this it is easy to see that the induction cases for the sibling axes
(e.g. $\rho=\LTLdiamond_{\ns} \psi$) go through.

This completes the proof of the claim.
Iterating the claim gives the proof of the first part of the theorem.

%For the second part, we can use a simpler pruning algorithm to shrink the branching
%-- for any node, consider the subformula types of its children. We can choose
%one child representing each type, and remove the rest.
%A standard inductive argument (exploiting the absence
%of sibling axes) shows that this preserves the truth of subformulas.

\section*{Proof of Theorem \ref{desconlyuar}}

Recall the statement:

\emph{The satisfiability of $\fotwo[\vdesc]$ with the 
unary alphabet restriction is $\nexptime$-hard.  }
\begin{proof}
We make use of a standard $\nexptime$-complete problem, tiling an exponential
sized grid \cite{tilepaper}.

The input consists of a number $n$ (in unary),
a set $C = \{1, \ldots, k\}$ of  colours, and a vertical and horizontal constraint 
$V,H \subset C \times C$.
A tiling is a mapping $f: \{1, 2, \ldots 2^n\} \times \{1, 2, \ldots 2^n\} \rightarrow C$, and a solution to the tiling
problem consists of a tiling
such that the vertical and horizontal constraints are satisfied.

Our formula will have in its signature
predicates 
\[
\zerox_1,\onex_1, \ldots, \zerox_n, \onex_n, \zeroy_1,\oney_1, \ldots, \zeroy_n, \oney_n
\]
 representing bits in the binary representation of the $x$- and $y$-coordinates
of a grid position, along with predicates $C_1 \ldots C_k$ for the colours, and finally a predicate $r$ for the root.
We code a tiling $f$ by a tree consisting of branches  of depth $2n+2$ for each grid position $ \{1, 2, \ldots 2^n\} \times \{1, 2, \ldots 2^n\}$.
If $f(x,y)=c$ then the branch will consist of a root, followed by  $n$ nodes, where the $i^{th}$ is labelled with $\zerox_i$ if the $i^{th}$ bit
of $x$ is $0$ and is labelled with $\onex_i$ otherwise. The branch will then have $n$ nodes coding the $y$-coordinate, labelled with
$\zeroy_i$ or $\oney_i$, and finally a leaf labelled with $c$.
Our $\fotwo[\vdesc]$ formula $\varphi$ will
%ensures that a tree that makes $\varphi$ satisfiable 
describe the encoding of a~valid $T$-tiling $f$.
It will include conjuncts enforcing the shape above:
\begin{compactitem}
\item There is a node with no ancestors labelled $r$, and this node
has a descendant labelled with $\zerox_1$ and another descendant labelled $\onex_1$.
\item Any node  with label $\zerox_i$ or $\onex_i$ for $i<n$ has a descendant labelled with  $\zerox_{i+1}$ and another with $\onex_{i+1}$,  
 such a node has no descendants labelled with $\zerox_j, \onex_j$ for $j<i$.
\item  Any node   with label $\zerox_n$ or $\onex_n$ has descendants labelled with  $\zeroy_{1}$ and another with $\oney_1$,  and 
 has no descendants labelled with $\zerox_j, \onex_j$ for $j<n$.
\item Any node  with label $\zeroy_i$ or $\oney_i$ for $i<n$ has descendants labelled with  $\zeroy_{i+1}$ and another with $\oney_{i+1}$,  and
 all its descendants are labelled with $\zerox_j, \onex_j$ for $j \geq i$ or with $c \in C$.
\item   For any node   with label $\zeroy_n$ or $\oney_n$, there is some $c \in C$   such that
$n$ has a descendant labelled  $c$ and no descendants with labels other than $c$.
\item Nodes labelled with $c \in C$ are leaves.
\end{compactitem} 

One can then write a formula $\text{SAME-X}(x,y)$ that checks whether two leaf nodes have the same $x$-coordinate:
\[
 \text{SAME-X}(x,y) = \bigwedge_i ((\exists y \, y \mathrel{\ancof} x \wedge \zerox_i(y) ) \leftrightarrow
(\exists x \, x \mathrel{\ancof} y \wedge \zerox_i(x) ))
\]

In the same way we can define $\text{SAME-Y}(x,y)$ to check whether two nodes agree on their $y$-coordinate, and
$\text{PLUS-X}(x,y)$,  $\text{PLUS-Y}(x,y)$ to check whether two nodes represent consecutive $x$- and $y$-coordinates, respectively.

The formulas above still allow the possibility of   many branches with the same co-ordinates but different colors, but
this can be enforced by the following formula, where $\text{LEAF}(x)$ states that $x$ is a leaf:
\[
\forall x \, \forall y \, (\text{LEAF}(x) \wedge \text{LEAF}(y) \wedge \text{SAME-X}(x,y) \wedge c(x)) \rightarrow c(y)
\]
The vertical and horizontal constraints can be enforced in the usual way given the formulas described above. For example:
\[
\forall x \, \forall y \, (\text{LEAF}(x) \wedge \text{LEAF}(y) \wedge \text{SAME-X}(x,y) \wedge \text{PLUS-Y}(x,y) \wedge c(x)) \rightarrow 
\bigvee_{(c,c') \in V} c'(y)
\]
Conjoining these sentences gives an $\fotwo[\vdesc]$ sentence  that holds on UAR trees iff a tiling exists.
\end{proof}

\section*{Proof of the polynomial alternation bound (Lemma \ref{lem:weistree}) }

Recall the statement of Lemma~\ref{lem:weistree}:

\emph{ Consider an $\fotwo[\vdesc]$ formula $\psi$ over unary
  predicates in $\Sigma$, and a~tree~$t$ satisfying the UAR.  For any
  symbol $a \in t$, and any root-to-leaf path $p=p_1 \ldots
  p_{max(p)}$ in $t$, the set $p(\psi,a):=\{i \mid  t, p_i \models \psi
  \wedge a(x)\}$ is made up of at most $ |\psi|^2$ $a$-intervals
  (i.e., intervals in the set $\{ i \mid t,p_i \models a(x)\}$.)}

The result relies on the following combinatorial lemma, which is
adapted from the argument in Lemma 2.1.10 of Weis \cite{WeisPhd}.
Analogously to the terminology above, given a word $w =w_1 \ldots
w_{max(w)}$ and a symbol $a$, by an $a$-interval we mean an interval 
in the set of positions in $w$ that have label $a$.

\begin{lemma}  \label{lem:count}
Consider a word $w$, a symbol $a$, formulas $\varphi_i(x):i \leq r$, and
$L,U$  functions that assign each boolean valuation  of the 
$\varphi_i(x)$ to positions of $w$.
Let $\beta$ be a positive boolean combination in
propositions $P_1 \ldots P_j$ and consider
the set 
\begin{align*}
J(w):=\{j \in w \mid w(j)=a \wedge (w, j) \models & \beta(\varphi_1, \ldots \varphi_r) \wedge\\ & (j \geq L(\comb(j)) \vee j<U(\comb(j)))
 \}
\end{align*}
where $\comb(j)$ is the boolean valuation of $\varphi_i:i \leq r$
induced by $j$ in $w$.  Suppose that for each $i \leq r$ the set of
position of $w$ labelled with $a$ satisfying $\varphi_i$ consists of
at most $|\varphi_i|^2$ $a$-intervals. Then the number of endpoints of
$a$-intervals comprising $J(w)$ is at most $4+2(\Sigma_i
|\varphi_i|)^2$.
\end{lemma}

We first show how Lemma \ref{lem:weistree} follows from Lemma
\ref{lem:count}. We proceed by induction. The base step follows using
the UAR, since for the predicate $b(x)$ the set $p(b,a)$ is either
empty or a single $a$-interval.  The cases for the boolean operations
are routine.

In the induction step for existential quantification, we consider a
formula $\psi(x)= \exists y \delta(x,y)$, where $\delta(x,y)$ is:
\[
\beta(x ~ \descof ~ y, x=y, x ~ \ancesof ~ y, x ~\incomp ~ y, \varphi_1, \ldots \varphi_r, \rho_1, \ldots \rho_s)
\]

We can assume $\beta$ is normalized to be a disjunction
of formulas $\beta_\descof$, $\beta_\ancesof$, $\beta_\incomp$, $\beta_=$,
where  $\beta_\descof(x,y)$ implies $y  ~ \descof ~ x$, and similarly for the others.
Thus in turn $\psi$ is the disjunction of  
$\psi_\descof, \psi_\ancesof, \psi_\incomp, \psi_=$ where
$\psi_R$ existentially quantifies over $\beta_R$.

For a boolean valuation $\sigma$ of the $\varphi_i$'s, and for
a relation $R$ in $\descof$, $\ancesof$, $\incomp$, $=$, we
let $\delta(\sigma,R)(y)$ be the formula obtained from $\delta(x,y)$
by replacing all $\varphi_i(x)$ in
$\delta$ by true
or false according to $\sigma$, formula $R(x,y)$ by true, and all other binary formulas
by false.

Fixing a root-to-leaf path $p=p_1 \ldots p_{max(p)}$ in tree $t$ (that is, where 
$p_1$ is the root, $p_{max(p)}$ a leaf), and $\sigma$ a boolean valuation of the $\varphi_i$'s
let:
\begin{compactitem}
\item $L_\incomp(\sigma)$ represent the smallest $i$ such
that
\[\exists n \in t \cdot n ~ \incomp ~p_i  \wedge t,n \models \delta(\sigma, \incomp)(y) 
\]
\item $U_\descof (\sigma)$ represent the largest $i$ such
that
\[\exists n \in t \cdot n ~ \descof ~ p_i  \wedge t,n \models \delta(\sigma, \descof )(y) 
\]
\item $L_\ancesof (\sigma)$ represent the smallest $i$ such
that
\[\exists n \in t \cdot n ~ \ancesof ~ p_i  \wedge t,n \models \delta(\sigma, \ancesof )(y) 
\]
\end{compactitem}

Unwinding the definitions, we can check that a node $p_j$ in the path $p$ within $t$ satisfies $\psi$
exactly when, letting $\sigma(j)$ be the boolean
valuation of the $\varphi_i$'s such that $t,p_j \models \varphi_i(x)$, we have either:
\begin{compactitem}
\item 
$j \leq U_\descof(\sigma(j))$ (thus $p_j$ 
has a witness to $\delta(\sigma(j), \descof)$, and hence a~witness to $\psi$ which is  a descendant).
\item  $j \geq L_\incomp(\sigma(j))$ (thus $p_j$ has a witness to $\psi$
that is incomparable to it).
\item  $j \geq L_\ancesof(\sigma(j))$ ($p_j$ has a witness to $\psi$
which is  an ancestor).
\item $t,p_i \models \psi_=(x)$, where $\psi_=$ is defined above.
\end{compactitem}

Restricting attention to $\psi_\descof \vee \psi_\ancesof \vee \psi_\incomp$,
we can apply Lemma \ref{lem:count} above,
letting $L(\sigma)$ be the  max of $L_\incomp(\sigma)$ and
$L_\ancesof(\sigma)$ and $U(\sigma)$ be  $U_\descof (\sigma)+1$.

We thus get that the number of boundary points of $a$-intervals
comprising $p(\psi_\descof \vee \psi_\ancesof \vee \psi_\incomp,a)$ is
at most $4+2 (\Sigma_i |\varphi_i|)^2$.

The boundary points of $p(\psi_=,a)$ are those of the $p(\rho_i,a)$,
and applying the induction hypothesis to these,
we get a bound on the number of endpoints of intervals comprising $p(\psi,a)$ 
as 
\[ 4+2 (\Sigma_i |\varphi_i|)^2 + 2 \Sigma_i |\rho_i|^2 \]
which is bounded by
$2 \cdot |\psi|^2$. Thus the number of intervals is bounded
by $ |\psi|^2$.
This completes the proof of Lemma \ref{lem:weistree}.

We now proceed to the proof of Lemma \ref{lem:count}.

We follow the approach of Lemma 2.1.10 of \cite{WeisPhd} and focus on
the modifications of the two main claims used in the proof of that
lemma.  For a formula $\psi(x)$ and letter $a$, let $w(\psi,a) = \{
i \in w : w,i \models \psi(x)\wedge a(x)\}$.

For $u \leq r$, let $F_u$ be the set of left boundaries of
$a$-intervals that comprise $w(\varphi_u,a)$, and let $G_u$ be the set
of right interval boundaries, where (by convention) we take the
decomposition into $a$-intervals of $w(\varphi_u,a)$ to be such that the
boundary points are labelled with $a$, the right (upper) boundary is
not part of $w(\varphi_u,a)$ but the left boundary is in $w(\varphi_u,a)$.
Let $F$ and $G$ be the total set of left and right interval boundaries
of $S$, and let $H = F \cup G \cup \{1, ||w||+1\}$.

Consider each interval $I$ defined by two consecutive elements of
$H$. The truth values of the $\varphi_i$ are constant on such an
interval, thus the truth value of $\varphi$ on positions $j$ in
this interval is determined by where $j$ is relative to $L(\comb(j))$
and $U(\comb(j))$.  Let $C$ be $H$ unioned with all points of the form
$L(\comb(j))+1$ or $U(\comb(j))$.

For a right (upper) interval boundary $d$ in $H$, we let $q(d)$ be the
point $L(\comb(j))+1$ for $j$ in the interval (all such points agree on
$\comb(j)$) to the left of $d$, if such a point exists; $q(d)$ is
undefined otherwise.  For a left (lower) interval boundary $c$ in $H$,
we let $p(c)$ be the point $U(\comb(j))$ to the right of $c$ within the
interval, if it exists, and let $p(c)$ be undefined otherwise.  We let
$P(c)=p(c)$ exactly when $p(c)$ is a right boundary point of
$J(w)$ -- that is, an $a$-labelled position lying outside of the
set, with the $a$-position immediately below it lying in the set. Let
$p(c)$ be undefined otherwise.  Similarly let $Q(c)=q(c)$ when $q(c)$
is a left boundary point of $J(w)$.

Let $F_{\bar{u}}$ be the union over all $F_v$ with $v \neq u$, and
define $G_{\bar{u}}$ analogously.

\begin{claim} Given $c$ and $d$ consecutive 
interval boundaries from $F_{\bar{u}}$,
there is at most one $i \in F_u \cap [c,d)$ with
$P(i) \neq \emptyset$.
\end{claim}

\begin{proof}
Suppose there is $i \in F_u \cap [c,d)$ with $P(i) \neq \emptyset$ and
consider another $j \in F_u \cap [c,d)$ with $j<i$.  Since the
interval $[c,d)$ contains no left interval boundaries besides the ones
from $F_u$, and since $i$ and $j$ are both in $F_u$, and hence are
both in $w(\varphi_u,a)$, we conclude that every $\varphi_k:k\leq r$
that holds in the interval starting from $i$ also holds at the
interval starting from $j$.  Thus $\comb(j)=\comb(i)$.  If $p(j)$ is a
right boundary point of $J(w)$, it must be that the positions
immediately below it are in the set $J(w)$, and thus these
positions must satisfy $x<U(\comb(x))$. Once truth values for the
$\varphi_k:k \leq r$ are fixed (and hence $\comb(x)$ is fixed), the
positions satisfying $x<U(\comb(x))$ are closed downwards.  Note that
$i<p(i)$, by definition of $p(i)$, and therefore we must have that $i$
and $j$ both satisfy $x<U(\comb(x))$. Combining with the fact that $i$
and $j$ agree on $\varphi_k:k \leq r$, we see that the interval above
$j$ agrees on $J(w)$ with the interval above $i$, and thus $P(j)$
must be empty.
\end{proof}

Let $C(i)$ be the set of boundary points contributed by $i$: namely $P(i)$ if it exists, $Q(i)$
if it exists, and also $i$ if it is a boundary point of $J(w)$.

\begin{claim}  Given $c$ and $d$ consecutive interval boundaries from $F_{\bar{u}}$,
and $i \in F_u \cap [c,d)$ with $i \not \in G$, $Q(i) \neq \emptyset$.
%and $Q(i) \subseteq C(i)\hat{C}_F(u,i)$
Then  we
have $i \not\in C(i)$.
%\hat{C}_F(u,i)$.
\end{claim}

\begin{proof}
Fix $c,d,i$ as in the claim.  Since $i \not \in G$, $i$ is not a right
interval boundary of any set $p(\varphi_j,a)$, and therefore the
$\varphi_j$ that are true at the interval ending at $i$ are also true
at the interval starting at $i$. Furthermore $Q(i) \neq \emptyset$
implies that $L(\comb(x))<x$ holds for $x$ above $Q(i)$, and thus will
hold for all $a$-labelled positions sharing $\comb(i)$ above $i$.  Thus
$i$ cannot be a boundary point for $J(w)$, and therefore
$i \not \in C(i)$.
\end{proof}

The rest of the argument follows that in \cite{WeisPhd} precisely.

The above two claims imply that for every $i \in F_u \cap [c,d)-G$ except
possibly one element, $C(i)$ is either empty,
contains the single element
$Q(i)$, or contains only $i$. At the one exceptional element $C(i)$ could consist of at most two
elements,  $P(i)$ and either $Q(i)$ or $i$ (but not both, by the second claim).

Therefore, 
$\bigcup_{i \in F_u \cap [c,d)-G}$ has at most $|F_u \cap [c,d)|+1$ elements.
Unioning over all intervals $[c,d)$ we get
\[
\Sigma_{i \in F_u-G} |C(i)| \leq 
 \Sigma _{c \in F_{\bar{u}}} (|F_u \cap [c,d)|+1)
= |F_{\bar{u}}| + |F_u|
\]
Using again the fact that each $C(i)$ contains at most two elements (see above), we also
know
$\Sigma_{i \in F_u-G} |C(i)| \leq 2 \cdot |F_u|$, and thus:
\[
\Sigma_{i \in F_u-G} |C(i)| \leq |F_u| + \mathit{min} \{ |F_u|, |F_{\bar{u}}|\}
\]
Since for each $j$, the number of intervals, and hence
the number of left endpoints of intervals,  is assumed to be at most $|\varphi_j|^2$, and using
that the sum of squares is less than the square of a sum we get:
\begin{align*}
\Sigma_{i \in F_u-G} |C(i)| \leq |\varphi_u|^2 +\mathit{min}\{  |\varphi_u|^2, (\Sigma_{i \neq u} |\varphi_i|)^2\} \\
\leq |\varphi_u|^2 + |\varphi_u| \cdot \mathit{min}\{ | \varphi_u| , \Sigma_{i \neq u} |\varphi_i| \}\\
\leq  |\varphi_u|^2 + |\varphi_u| \cdot  \Sigma_{i \neq u} |\varphi_i| \\
=  |\varphi_u| \cdot \Sigma_{i} | \varphi_i|
\end{align*}
By a symmetric argument we get
\[
\Sigma_{i \in G_u-F} ~ |C(i)| \leq |\varphi_u| \cdot \Sigma_{i} |\varphi_i|
\]

Now the total number of boundary points for $J(w)$ is at most the
endpoints of the path, the highest value of $U$ and the lowest value
of $L$, plus the union over $i$ of $C(i)$. Thus we have that the total
number is at most:
\[
4+ \Sigma_u 2 \cdot |\varphi_u| \cdot  \Sigma_{i} |\varphi_i|
\leq 4 + 2 \cdot (\Sigma_i |\varphi_i|)^2
\]
This completes the proof of Lemma \ref{lem:count}.

%\subsection{$\fotwo[\vchild]$}

\section*{Proof of Theorem \ref{fod}}

Recall the statement:

\emph{The satisfiability problem for $\fotwo[\vchild]$ is $\nexptime$-hard,
even with the unary alphabet restriction.}

\begin{proof}
Clearly, the UAR has no impact, since $n$ predicates on a single
node can be simulated by considering the labels of the $n$ nearest ancestors.

We reduce from tiling a $2^n$ by $2^n$ grid with tiles $T_1 \ldots T_m$ in such
a way to satisfy a given vertical constraint $V$ and horizontal constraint $H$.
We let $\Sigma_n$ be an alphabet with  symbols $\zerox_1, \onex_1, \ldots ,\zerox_{n}$, $\onex_n, \zeroy_1$, $\oney_1 \ldots, \zeroy_{n}$, $\oney_n, T_1 \ldots T_m$.
Consider trees in which: nodes at level 
 $i \leq n$ are labelled with 
  $\zerox_i$ or $\onex_i$,  each node of level $i \leq n-1$ has both a $\zerox_{i+1}$ and an $\onex_{i+1}$
child. 
Similarly nodes at level 
 $n+1 \leq i \leq 2n$ are labelled with 
  $\zeroy_i$ or $\oney_i$. Each node of level $n$ has both $\zeroy_{1}$ and an~$\oney_{1}$ child,
   each node of level $n+1 \leq i \leq 2n-1$ has both a $\zeroy_{i+1}$ and an~$\oney_{i+1}$
child.

Finally, each node of level $2n$ has a single child
labelled with one of the $T_i$.  Such trees represent a tiling of the  grid. It is easy
to write an $\fotwo[\vchild]$ formula describing such trees, and also requiring that
the horizontal and vertical constraints are satisfied.
\end{proof}

\section*{Completion of the proof of Theorem \ref{fodescup}}

Recall the statement:

\emph{The satisfiability problem for $\fotwo[\vdesc]$ over ranked schemas
is in $\nexptime$, and is thus $\nexp$-complete.}

We first prove the key lemma, Lemma
\ref{lem:collapsevdesc}. Recall that in this lemma, we replace
node $n$ by node $n'$, where $n$ and $n'$  are not in the protected witness set $W$ and
share the same $\varphi$-type, the same set of ancestor $\varphi$-types,
and the same set of selected descendant $\varphi$-types. The lemma then claims:

For all $m \in T_1$ the 
one-variable subformulas of
$\varphi$ satisfied by $m$ in $t$ are the same as those
satisfied by $f(m)$ in $t'$.
Moreover, for every node $m'$ in $T_2$,
the one-variable subformulas of
$\varphi$ satisfied by $m'$ in $t'$ are the same as those
satisfied by $f^{-1}(m')$ in $t$.

We prove both parts of the lemma
by simultaneous
induction on the structure of the formula, where the case of atomic propositions and the case of boolean combinations are trivial. The only interesting case is for subformulas $\rho = \exists y \beta(x,y)$. 

We first note the following key property of the
witness set $W$: 
For nodes $m$ of $t$, if there is a $w$ incomparable to $m$ such that
$t,m, w \models \beta(x,y)$, then there is such a $w$ in $W$.

To prove this,
fix $m$ and $w$ such that the hypothesis holds.  Let $w_\tau$ be the basic
global witness for the $\varphi$-type of $w$. If $w_\tau$ is incomparable to $m$, then
$w_\tau$ has the required property. If $w_\tau$ is a descendant of $m$, then
we would have thrown in the necessary $w$ into $W$ as an
incomparable global witness for $m$.
If $w_\tau$ is an ancestor of $m$ or equal to $m$, we would have thrown in the necessary
$w$ into $W$ as an incomparable global witness for $w_\tau$.

We begin by comparing formulas $\rho$ between a 
node $m$ of the old tree (i.e. $m \in T_1$) and the same node considered
in the new tree. 
We first consider the case where $\rho$ 
%\rasto{shouldn't these two $\varphi$ be $\psi$?}
holds at $m$ in $t$, and show that $\rho$ remains true at its image $f(m)$
in $t'$.
\begin{compactitem}
\item If the  witness of the truth of
$\psi$ was $m$ or its ancestor, then these
are also in $T_1$, and thus are preserved under the mapping,
so by induction they (i.e. their image under
$f$) can serve as a witness in $t'$.
\item Suppose there is a witness $w$ that is
 neither $m$, nor an ancestor of $m$, nor a~descendant of $m$. 
By the key property of $W$, there is a witness $w'$ in the set $W$ that is also
incomparable to $m$, and has  the same $\varphi$-type as $w$. This can be used as a witness.
\item The last possibility is that some
of the witnesses are  descendants. If at least one of these is
not in $\subtree(t,n)$,
then it is  preserved and can be used as a witness.
 Otherwise, the witness must be in $\subtree(t,n)$.
If $n$ itself was a witness, then since it was replaced by an
$n'$ such that $\tp_\varphi(n') = \tp_\varphi(n')$ we can use the copy of $n'$
as a witness, by induction. 
On the other hand, if there was a descendant of $n$ which was a witness, then there would have been a 
witness $w''$ such that $\tp_\varphi(w'') \in \sdesctype(n)$. Since $\sdesctype(n)=\sdesctype(n')$ 
we would be able to find a~witness with the appropriate $\varphi$-type in a copy of the subtree rooted at $n'$.
\end{compactitem}

We now consider the case where $\rho$ holds at a node
$m'$  that is the image of a node $m$ in $\subtree(t,n')$ under
the overwriting operation,
and aim to show that $\rho$ holds at $m$. Note that once this is shown,
the other direction of the if and only if for nodes in $T_1$ follows
easily by induction.
So fix such $m'$ and $m$.
The only non-trivial case is for $m'$ being a copy of $n'$,
with the witness being its ancestor. Here we can use as a witness
one of the ancestors of $n$, because $\anctype(n) = \anctype(n')$.

This completes the proof of Lemma \ref{lem:collapsevdesc}.
The argument
for Theorem \ref{fodescup} for UAR trees proceeds by repeatedly updating while such nodes
are available. The process terminates, as argued in the body of the paper.

The extension for ranked schemas follows along the same lines, but in order to collapse nodes $n$ and $n'$,
we require in addition that the tree automaton $A$ reaches the same state at $n$ and $n'$.

\section*{Proof of Theorem  \ref{rankednochilduar}}
Recall the statement:

\emph{There are  $\fotwo[\vdesc]$ formulas $\varphi_n$ of size $O(n)$ that are satisfiable over
UAR binary trees, where the minimum depth of satisfying binary UAR trees grows as $2^n$.}
\begin{proof}
We let $\Sigma_n$ consist of  $\{ b, s\} \cup \{a_i: i \leq n\}$.

We consider trees in which:
\begin{compactitem}
\item  the root is labelled $b$
\item nodes labelled  $b$ are always comparable via descendant
\item nodes labelled  $s$ are never comparable via descendant
\item  every ancestor of a $b$-labelled node is labelled $b$
\item every ancestor of an $s$-labelled node is labelled $b$
\item descendants of $s$-labelled nodes  can be labelled with any of the $a_i$ (but not with $b$)
\end{compactitem}

These conditions can easily be enforced by an $\fotwo[\vdesc]$ formula.

In such trees the $b$-labelled nodes must go down a single branch, with $s$-labelled nodes
splitting off on a separate branch. See Figure \ref{fig:bigtree}.
We  now let
$\psi_i: i \leq n$ be the formula that holds at an $s$-labelled node
if it has a descendant $a_i$.
Note that any combination of the $\psi_i$ are consistent, and the set of $\psi_i$ that
hold of an~$s$-labelled node can thus be considered an $n$-bit address for the $s$-node.
We can write a formula $\varphi_n$ that asserts that
\begin{inparaenum}  \item the constraint on the shape of the tree above holds
\item   there is an
$s$-node with address $0^n$
\item  for every $s$-labelled node with address $a$ not
equal to $1^n$, there is an $s$-labelled
 node whose bit address is the successor of $a$. 
\end{inparaenum}
A binary tree satisfying $\phi_n$ must have exponential depth. See Figure \ref{fig:bigtree} for
an example.
\end{proof}

\tikzset{
gluon/.style={draw=black, line width = 0.4mm}
}
\begin{figure}[h!]
\begin{center}
\scalebox{0.95}{
\includegraphics{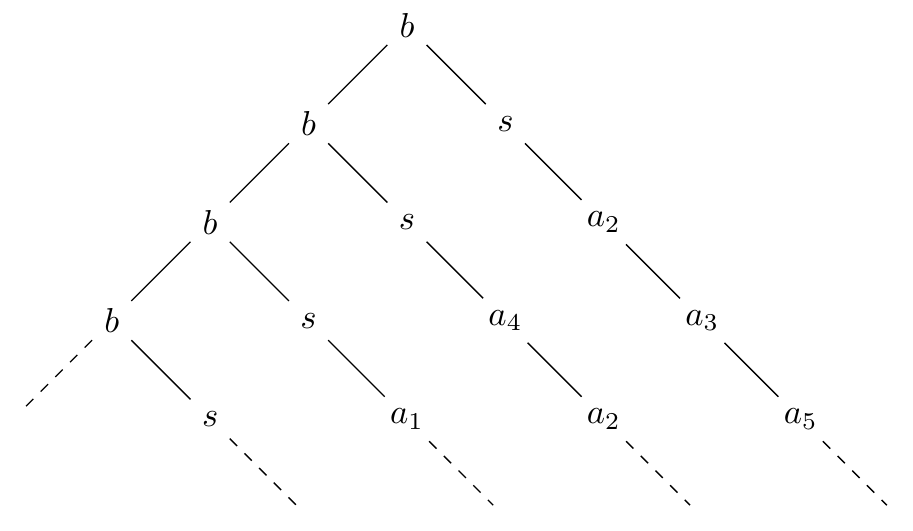}
}
\end{center}
\caption{An example model of exponential depth for $\fotwo[\vdesc]$ formula in ranked case}
\label{fig:bigtree}
\end{figure}

\section*{Details for the proof of Theorem \ref{fodup}}

%\subsection{$\fotwo[\vchild]$}

Recall the statement:

\emph{The satisfiability problem for $\fotwo[\vnodesc]$, and the satisfiability problem with respect
to a rank schema,
are in
$\nexptime$, and hence are $\nexptime$-complete.}

We give the details for satisfiability first.
By Lemma \ref{depthlemma}
 we know that 
 a $\fotwo[\vnodesc]$ formula $\varphi$ which is satisfied over trees is satisfied by a tree $t$ of 
depth at most exponential in $\varphi$. We also can bound the outdegree of nodes by an exponential.

%Recall our procedure: 
For each $\varphi$-type that is satisfied in $t$,
 choose a satisfier and include 
it along with all its ancestors in a set $W$: these are the basic witnesses.
Then throw in all children of basic witnesses.

Thus the size of $W$ is at most exponential. 
Now we transform $t$ to another tree $t'$ such that $t' \models \varphi$ and $t'$ has only exponentially many different subtrees.

Recall that our update procedure looks for
if there are nodes $n, n'$ in $t$  such that
\begin{inparaenum}
\item $n,n' \not \in W$ 
\item $\subtree(t,n) \prec \subtree(t,n')$ is not isomorphic to the~subtree rooted at $n'$
%\item $n$ and $n'$ have the same $\varphi$-type, and the parent of $n$ has the same $\varphi$-type as the parent of $n'$
\item $\tp_{\varphi}(n)=\tp_{\varphi}(n')$  and  $\tp_{\varphi}(\text{parent}(n))=\tp_{\varphi}(\text{parent}(n'))$
\end{inparaenum}
then let $t'=\up(t)$ be obtained
by choosing such $n$ and $n'$
and  applying the~collapse operation that replaces the subtree of $n$ by that of $n'$.
%replacing the subtree rooted
%at  $n_i$ by the subtree rooted at $n_j$.
%Note that $t'$ also contains a copy of $W$.

Let $T_1$ be
the nodes that were not in $\subtree(t,n)$,
and for any node $m \in T_1$ let $f(m)$
denote the same node viewed in $t'$.
Let
$T_2$ denote the nodes in $t'$ that are images
of a node in $\subtree(t,n')$
 under the replacement. For each $m \in T_2$,
let $f^{-1}(m)$ denote the node in $t$ from which  it derives.

We claim the following:
\begin{lemma}
For all $m \in T_1$ the $\varphi$-type of $m$ in $t$
is the same as the $\varphi$-type of
 $f(m)$ in $t'$.
Moreover, for every node $m'$ in $T_2$, 
the $\varphi$-type
of $m'$ in $t'$ is the same as the $\varphi$-type of
 $f^{-1}(m)$ in $t$.
\end{lemma}

Applying the lemma above to the root of
$t$, which is necessarily in $T_1$, it  follows that the truth
of the sentence $\varphi$ is preserved by this operation.

\begin{proof}
We prove both parts of the lemma
by simultaneous
induction on the structure of the formula, where the case of atomic propositions and the case of boolean combinations are trivial. The only interesting case is for subformulas $\psi = \exists y \beta(x,y)$.

We begin by considering formula $\psi$ at node $m \in T_1$.
We first consider the case where $\varphi$ holds at $m$.
\begin{compactitem} 
\item If the  witness of the truth of
$\psi$ was $m$ or its parent, then these
are also in $T_1$, and thus are preserved under the mapping,
so by induction they (i.e. their image under
$f$) can served as a witness in $t'$.
\item Similarly, if the witness was a sibling of $m$, then it can serve
as a witness in $t'$, since the collapse map does not impact the sibling relations.
%check this!
\item If all witnesses are neither a parent nor a child of $m$,
then take one such witness $w$
and  an element  $w'$ in $W$ that realizes the same $\varphi$-type
 as $w$.  $w'$ must be neither  a parent or a child of $m$ (since if it were a parent, $m$ would have
been a child witness, and hence protected).
Thus $w'$ can be used as a~witness.
\item The last possibility is that some
of the witnesses are  children. If at least one of these is
not in $\subtree(t,n)$,
then it is  preserved and can be used as a~witness.
 Otherwise, 
$n$ itself must be a witness. It was replaced by an
$n'$ such that $\tp_\varphi(n) = \tp_\varphi(n')$ so the copy of $n'$
can be used  as a witness, by induction.
\end{compactitem}

We now consider the case where $\psi$ holds at a node $m' \in T_2$
that is the image of a node $m \in T$,
and aim to show $\psi$ holds at $m$. 
The only non-trivial case is for $m'$ being the image of $n'$,
with the witness being its parent. Here we can use as a~witness 
the parent of $n$, because $\tp_\varphi$ of the parent of $n$ is
the same as $\tp_\varphi$ of the parent of $n'$.
\end{proof}

We now iterate the procedure $t_{i+1}~ := ~ \up(t_i)$, until
no more updates are possible. Since $t_{i+1} \prec t_i$,
the process must terminate.
The resulting tree will contain only exponentially many different subtrees. 
We can thus represent it as a~DAG, with one node for each subtree.

Thus we have shown that any satisfiable formula has an exponential-size
DAG that unfolds into a model of the formula. Given
such a DAG, we can check whether an $\fotwo$  formula holds
in polynomial time in the size of the DAG.
Thus we have
a $\nexptime$ algorithm for checking satisfiability.

The modification in the presence of a ranked schema is straightforward -- again
we show that there is an exponential-sized DAG.
Given a bottom-up tree-automaton, the modification
procedure $\up$  only
replaces $n$ by $n'$ if, in addition to the criteria above,
their subtrees reach the same state of $A$.
Clearly, the state of $A$ is also preserved by this replacement.
This  completes
the proof of Theorem \ref{fodup}.

\end{document}